\documentclass{article}
\usepackage[left=3cm,right=3cm,top=3cm,bottom=3cm]{geometry} 
\usepackage{amsmath,amssymb} 
\usepackage{amsthm}
\usepackage{tikz}
\usepackage{hyperref}
\usetikzlibrary{patterns}
\usetikzlibrary{intersections}
\usetikzlibrary{arrows,positioning}
\usepackage{graphicx}
\usepackage{pgfplots}
\pgfplotsset{compat=1.14}
\usepgfplotslibrary{fillbetween}
\usepackage{float}
\usepackage{fancyhdr}
\usepackage{caption}
\usepackage{subcaption}

\newcommand*{\rom}[1]{\expandafter\@slowromancap\romannumeral #1@}

\newcommand{\gettikzxy}[3]{%
  \tikz@scan@one@point\pgfutil@firstofone#1\relax
  \edef#2{\the\pgf@x}%
  \edef#3{\the\pgf@y}%
}

\newcommand{\sA}{\mathcal A}
\newcommand{\sB}{\mathcal B}
\newcommand{\sC}{\mathcal C}
\newcommand{\sD}{\mathcal D}
\newcommand{\sF}{\mathcal F}
\newcommand{\sG}{\mathcal G}
\newcommand{\sH}{\mathcal H}
\newcommand{\sI}{\mathcal I}
\newcommand{\sL}{\mathcal L}
\newcommand{\sM}{\mathcal M}

\newcommand{\sP}{\mathcal P}

\newcommand{\sR}{\mathcal R}
\newcommand{\sS}{\mathcal S}
\newcommand{\sT}{\mathcal T}

\newcommand{\sW}{\mathcal W}

\newcommand{\R}{\mathbb R}
\newcommand{\E}{\mathbb E}
\newcommand{\F}{\mathbb F}

\newcommand{\Prob}{\mathbb P}

\newtheorem{thm}{Theorem}

\newtheorem{lem}{Lemma}

\newtheorem{rem}{Remark}
\newtheorem{ass}{Assumption}
\newtheorem{defn}{Definition}
\newtheorem{sass}{Standing Assumption}
\newtheorem{prob}{Problem}

\begin{document}
\title{Robust bounds for the American Put}
\author{David Hobson\thanks{E-mail: \textit{d.hobson@warwick.ac.uk}}\hspace{3mm}  and  \hspace{3mm}Dominykas Norgilas\thanks{E-mail: \textit{d.norgilas@warwick.ac.uk}}\\
Department of Statistics, University of Warwick\\Coventry CV4 7AL, UK}
\date{\today}
\maketitle

\begin{abstract}
We consider the problem of finding a model-free upper bound on the price of an American put given the prices of a family of European puts on the same underlying asset. Specifically we assume that the American put must be exercised at either $T_1$ or $T_2$ and that we know the prices of all vanilla European puts with these maturities. In this setting we find a model which is consistent with European put prices and an associated exercise time, for which the price of the American put is maximal. Moreover we derive the cheapest superhedge. The model associated with the highest price of the American put  is constructed from the left-curtain martingale coupling of Beiglb\"{o}ck and Juillet.
\end{abstract}

\section{Introduction}
\label{sec:intro}

This article is motivated by an attempt to understand the range of possible prices of an American put in a robust, or model-independent, framework. In our interpretation this means that we assume we are given today's prices of a family of European-style vanilla puts (for a continuum of strikes and for a discrete
set of maturities). The goal is to find the consistent model for the underlying for which the American put has the highest price, where by definition a model is consistent if the discounted price process is a martingale and if the model-based discounted expected values of European-put payoffs match the given prices of European puts.

This notion of model-independent or robust bounds on the prices of exotic options was introduced in Hobson~\cite{Hobson:98} in the context of lookback options, and has been applied several times since, see Brown et al.~\cite{BrownHobsonRogers:01} (barrier options), Cox and Ob{\l}{\'o}j~\cite{CoxObloj:11} (no-touch options), Hobson and Neuberger~\cite{HobsonNeuberger:12} and Hobson and Klimmek~\cite{HobsonKlimmek:15} (forward-start straddles), Carr and Lee~\cite{CarrLee:10} and Cox and Wang~\cite{CoxWang:13} (variance options), 
Stebegg~\cite{Stebegg:14} (Asian options) and the survey article Hobson~\cite{Hobson:survey}. The principal idea is that the prices of the vanilla European puts determine the marginal distributions of the price process at the traded maturities (but not the joint distributions) and that these distributional requirements, coupled with the martingale property, place meaningful and useful restrictions on the class of consistent models. These restrictions lead to bounds on the expected payoffs of path-dependent functionals, or equivalently bounds on the prices of exotic options.

In addition to the pricing problem there is a related dual or hedging problem. In the dual problem the aim is to construct a static portfolio of European put options and a dynamic discrete-time hedge in the underlying which combine to form a superhedge (pathwise over a suitable class of candidate price paths) for the exotic option. The value of the dual problem is the cost of the cheapest superhedge. There is a growing literature, beginning with Beiglb\"{o}ck et al.~\cite{BeiglbockHenryLaborderePenkner:13} for discrete-time problems, and Galichon et al.~\cite{GalichonHenryLabordereTouzi:14} in continuous time, which aims to explain how to formulate the problem in such a way that there is no duality gap, i.e. the highest model-based price is equal to the cheapest superhedge, either for specific derivatives, or in general.

Many of the early papers on robust hedging exploited a link with the Skorokhod embedding problem (Skorokhod~\cite{Skorokhod:65}).
For example, in the study of the lookback option in Hobson~\cite{Hobson:98} the consistent model which achieves the highest lookback price is constructed from the Az\'{e}ma-Yor~\cite{AzemaYor:79} solution of the Skorokhod embedding problem.  More recently, Beiglb\"{o}ck et al.~\cite{BeiglbockHenryLaborderePenkner:13}
(see also Dolinsky and Soner~\cite{DolinskySoner:14} and Touzi~\cite{Touzi:13}) have championed the connection between robust hedging problems and martingale optimal transport. In this paper we will make use of the left-curtain martingale coupling introduced by Beiglb\"{o}ck and Juillet~\cite{BeiglbockJuillet:16}, and developed by Henry-Labord\`{e}re and Touzi~\cite{HenryLabordereTouzi:16} and Beiglb\"{o}ck et al.~\cite{BeiglbockHenryLabordereTouzi:17}.

The study of American style claims in the robust framework was initiated by Neuberger~\cite{Neuberger:07}, see also Hobson and Neuberger~\cite{HobsonNeuberger:17}, Bayraktar and Zhou~\cite{BayraktarZhou:16} and Aksamit et al.~\cite{AksamitDengOblojTan:17}. (There is also a paper by Cox and Hoeggerl~\cite{CoxHoeggerl:16} which asks about the possible shapes of the price of an American put, considered as a function of strike, given the prices of co-maturing European puts.) The main innovation of this paper is that rather than focussing on general American payoffs and proving that the pricing (primal) problem and the dual (hedging) problem have the same value, we focus explicitly on American puts and try to say as much as possible about the structure of the consistent price process for which the model-based American put price is maximised, and the structure of the cheapest superhedge.

Mathematically, it will turn out that our problem can be cast as follows. Let $\mu$ and $\nu$ be a pair of probability measures which are increasing in convex order and therefore necessarily have the same mean $\bar{\mu}$. A standing assumption in this paper will be that $\mu$ is continuous (or equivalently, $\mu$ has no atoms). Let $\hat{\Pi}_M(\mu,\nu)$ be the set of martingale couplings (which are often alternatively called martingale transports) between $\mu$ and $\nu$ and let $K_1>K_2$ be a pair of fixed constants. The problem we consider is to find
\begin{equation}
\label{eq:primalcts}
 \sup_{\pi \in \hat{\Pi}_M(\mu,\nu)} \sup_{B \in \sB(\R)} \E^{\sL(M_1,M_2) \sim \pi}  \left[ (K_1 - M_1)^+ I_{  \{  M_1 \in B \} } + (K_2 - M_2)^+ I_{ \{ M_1 \notin B \} }  \right] ,
\end{equation}
where $M=(\bar{\mu}, M_1, M_2)$ is a martingale with joint law $\Prob(M_1 \in dx, M_2 \in dy) = \pi(dx,dy)$ and $B$ is a Borel subset of $\R$.
In terms of the American put problem $M$ should be thought of as the discounted price of the underlying asset (to simplify notation we write $M_1 \equiv X$ and $M_2 \equiv Y$).  Further,  $K_1$ and $K_2$ are the discounted strikes of the put and $B$ represents the set of values of the discounted time-1 price of the underlying such that the option is exercised at time 1; otherwise the put is exercised at time 2.  Then \eqref{eq:primalcts} represents the primal problem of finding the highest model-based expected payoff of the American put. See Section~\ref{ssec:putproblem}.

There is a corresponding dual or hedging problem of finding the cheapest superhedge based on static portfolios of European puts and a piecewise constant holding of the underlying asset, see Section~\ref{ssec:superhedge}.

Our main achievement is {\em  to exhibit the model and stopping rule which achieves the highest possible price for the American put, to exhibit the cheapest superhedge, and to show that the highest model-based price is equal to the cost of the cheapest superhedge}.

For fixed $\mu,\nu$ and $K_1 > K_2$ there is typically a family of optimal models. Fixing $\mu$ and $\nu$ but varying $K_1$ and $K_2$ it turns out that there is a model which is optimal for all $K_1$ and $K_2$ simultaneously. This model is related to the left-curtain coupling of Beiglb\"ock and Juillet~\cite{BeiglbockJuillet:16}. In particular, given $\mu\leq_{cx} \nu$ (with $\mu$ continuous),  Beiglb\"ock and Juillet~\cite{BeiglbockJuillet:16} prove that there exist functions $T_d$ and $T_u$ with $T_d(x) \leq x \leq T_u(x)$ such that $T_u$ is increasing and such that if $x<x'$ then $T_d(x') \notin (T_d(x), T_u(x))$, and such that there is $\pi \in \hat{\Pi}_M(\mu,\nu)$ which is concentrated on the graphs of $T_d$ and $T_u$. Under this martingale coupling $Y \in \{ T_d(X), T_u(X) \}$ and by the martingale property $\Prob(Y=T_d(X)|X) = \frac{T_u(X)-X}{T_u(X) - T_d(X)}$ (assuming not both $T_d(X) = X$ and $T_u(X)=X$).

In this paper we will concentrate on the case where $\mu$ is continuous.
Indeed, if $\mu$ has atoms then the situation becomes more delicate. On one hand, we must allow for a wider range of possible candidates for exercise determining sets $B$. On atoms of $X$ we may want to sometimes stop and sometimes continue, although we must still take stopping decisions which do not violate the martingale property of future price movements. On the other hand, the functions $T_d$, $T_u$ that characterises the left-curtain coupling become multi-valued on the points where $\mu$ has atoms. Then it is not clear how the optimal model can be identified. For these reasons we must extend our notion of a martingale coupling and generalise, in a useful fashion, the left-curtain martingale coupling of Beiglb\"{o}ck and Juillet~\cite{BeiglbockJuillet:16} to the case with atoms. The appropriate extension of the left-curtain coupling to the case with atoms in $\mu$ is discussed in a companion paper (\cite{HobsonNorgilas:18}); in this paper we focus on the financial aspects of our results, namely the application to the robust hedging of American puts.

The remainder of the paper is structured as follows. In the next section we formulate precisely our problem of finding the robust, model-independent price of an American put and explain how the problem can be transformed into \eqref{eq:primalcts} in the atom-free case. We also explain how the pricing problem is related to the dual problem of constructing the cheapest superhedge. In Section~\ref{sec:continuous} we assume that $\mu$ is continuous, and we show by studying a series of ever more complicated set-ups how to determine the best model and hedge. The constructions in this section make use of results on the left-curtain coupling of Beiglb\"ock and Juillet~\cite{BeiglbockJuillet:16} and Henry-Labord\`{e}re and Touzi~\cite{HenryLabordereTouzi:16}.

By weak duality the highest model price is bounded above by the cost of the cheapest superhedge. Hence, if on the one hand we can identify a consistent model and stopping rule and on the other a superhedge, such that the expected payoff in that model with that stopping rule is equal to the cost of the superhedge then we must have identified an optimal model and an optimal stopping rule together with an optimal hedging strategy. Moreover there is no duality gap. This is the strategy of our proofs. One feature of our analysis is that wherever possible we provide pictorial explanations and derivations of our results. In our view this approach helps bring insights which may be hidden under calculus-based approaches.

\section{Preliminaries and set-up}
\label{sec:setup}

\subsection{Measures and Convex order}
Given an integrable measure $\eta$ (not necessarily a probability measure) on $\R$ define $\bar{\eta} = \frac{\int_\R x \eta(dx)}{\int_\R \eta(dx)}$ to be the barycentre of $\eta$. Let $\sI_\eta$ with endpoints $\{ \ell_\eta, r_\eta \}$ be the smallest interval containing the support of $\eta$.
Define $P_\eta : \R \mapsto \R^+$ by $P_\eta(k) = \int_{-\infty}^k (k-x) \eta(dx)$.
Then $P_\eta$ is convex and increasing, and represents the discounted European put-price, expressed as a function of strike, if the discounted underlying has law $\eta$ at maturity. Further, $\{k : P_{\eta}(k) > \eta(\R)(k - \bar{\eta})^+ \} \subseteq \sI_\eta$.
Note that $P_{\eta}$ is related to the potential $U_\eta$ defined by $U_\eta(k) : = - \int_{\R} |k-x| \eta(dx)$ by
$P_\eta(k) = \frac{1}{2}(- U_\eta(k) + (k- \bar{\eta}) \eta(\R))$.

For any $c<d$ and a measure $\eta$ let $\eta_{c,d}$ be the measure given by $\eta_{c,d}(A)= \eta(A \cap (c,d))$. Let $\tilde{\eta}_{c,d} = \eta - \eta_{c,d}$.

Two measures $\eta$ and $\chi$ are in convex order and we write $\eta \leq_{cx} \chi$ if and only if $\eta(\R)= \chi(\R)$, $\bar{\eta}= \bar{\chi}$ and $P_{\eta}(k) \leq P_{\chi}(k)$ on $\R$. Necessarily we must have
$\ell_\chi \leq \ell_\eta \leq r_\eta \leq r_\chi$. 
Let $\hat{\Pi}_M(\eta,\chi)$ be the set of martingale couplings of $\eta$ and $\chi$. Then
\[  \hat{\Pi}_M(\eta,\chi) = \left\{ \pi \in \sP( \R^2) : \mbox{$\pi$ has first marginal $\eta$ and second marginal $\chi$} ; \mbox{\eqref{eq:martingalepi} holds} \right\}
\]
where $\sP(\R^2)$ is the set of probability measures on $\R^2$ and \eqref{eq:martingalepi} is the martingale condition
\begin{equation}
\int_{x \in B} \int_{y \in \R} y \pi(dx,dy) = \int_{x \in B} \int_{y \in \R} x \pi(dx,dy) = \int_B x \eta(dx) \hspace{10mm}
\mbox{$\forall$ Borel $B \subseteq \R$}.
\label{eq:martingalepi}
\end{equation}

For a pair of measures $\eta,\chi$ define $D = D_{\eta,\chi}: \R \mapsto \R^+$ by $D_{\eta,\chi}(k) = P_\chi(k) - P_\eta(k)$. Note that if $\eta,\chi$ have equal mass and equal barycentre then $\eta \leq_{cx} \chi$ is equivalent to $D \geq 0$ on $\R$.
Let $\sI_D = [\ell_D,r_D]$ be the smallest closed interval containing $\{ k : D_{\eta,\chi}(k)>0 \}$. If $\sI_D$ is such that $\sI_D \subset \sI_\chi$ then we must have $\eta=\chi$ on $[\ell_\chi, \ell_D) \cup (r_D, r_\chi]$.

The following lemma tells us that if $D_{\eta,\chi}(x)=0$ for some $x$ then in any martingale coupling of $\eta$ and $\chi$, no mass can cross $x$.
A proof can be found in Hobson~\cite[p254]{Hobson:98b}.

\begin{lem} Suppose $\eta$ and $\chi$ are probability measures with $\eta \leq_{cx} \chi$. Suppose $D(x)=0$. If $\pi \in \hat{\Pi}_M(\eta,\chi)$ then we have $\pi((-\infty,x),(x,\infty)) + \pi((x,\infty),(-\infty,x))=0$.
\label{lem:chacon}
\end{lem}

It follows from Lemma \ref{lem:chacon} that if there is a point $x$ in the interior of the interval $\sI_\eta$ such that $D_{\eta,\chi}(x)=0$ then we can separate the problem of constructing martingale couplings of $\eta$ to $\chi$ into a pair of subproblems involving mass to the left and right of $x$, respectively, always taking care to allocate mass of $\chi$ at $x$ appropriately. Indeed, if there are multiple $\{ x_j \}$ with $D_{\eta,\chi}(x_j)=0$ then we can divide the problem into a sequence of `irreducible' problems\footnote{The terminology `irreducible' is due to Beiglb\"{o}ck and Juillet~\cite{BeiglbockJuillet:16} although the idea of splitting a problem into separate components is also present in the earlier papers of Hobson~\cite{Hobson:98b} and Cox~\cite{Cox:08}.}, each taking place on an interval $\sI_i$ such that $D>0$ on the interior of $\sI_i$ and $D=0$ at the endpoints. All mass starting in a given interval is transported to a point in the same interval. However, in our setting, in addition to specifying a model (or equivalently a martingale coupling) we also need to specify a stopping rule, and this needs to be defined across all irreducible components simultaneously. For this reason we do not insist that $D>0$ on the interior of $\sI_\chi$, although this will be the case in the simple settings in which we build our solution.

\subsection[Financial model]{The financial model and model based prices for American puts}
\label{ssec:putproblem}
Suppose $\tilde{Z} = (\tilde{Z}_{T_i})_{i = 0,1,2}$ is the price of a financial security which pays no dividends, where $T_0=0$ is today's date. (In this section a superscript $\tilde{\cdot}$ denotes an undiscounted quantity.) Suppose interest rates are non-stochastic and positive. Let one unit of cash invested at time $T_0$ in a bank account paying the riskless rate be worth
$\tilde{B}_{T_i}$ at time $i$ for $i = 0,1,2$. Then $\tilde{B}_0=1$. Define $Z = ({Z}_{i})_{i = 0,1,2}$ by $Z_{i}= \tilde{Z}_{T_i}/\tilde{B}_{T_i}$ so that $Z$ is the discounted asset price with a simplified time-index $i= 0,1,2$. We assume that $Z_0$ is known at time 0.

Let $\Sigma$ be the set of stopping rules taking values in $\{ T_1, T_2 \}$ and let $\sT$ be the set of stopping rules taking values in $\{ 1,2 \}$.
Consider an American put with strike $\tilde{K}$ which may be exercised at $T_1$ or $T_2$ only. Define $K_i = \tilde{K}/\tilde{B}_{T_i}$. Under a fixed model the expected payoff of an American put under an exercise (stopping) rule $\sigma$ taking values in $\{ T_1, T_2 \}$ is given by $\E[\frac{1}{\tilde{B}_{\sigma}} (\tilde{K} - \tilde{Z}_{\sigma})^+]$ and the price of the American option (assuming exercise is only allowed at $T_1$ or $T_2$) is
\[ \sup_{\sigma \in \Sigma} \E \left[ \frac{1}{\tilde{B}_{\sigma}} (\tilde{K} - \tilde{Z}_{\sigma})^+ \right] = \sup_{\tau \in \sT} \E \left[ (K_\tau - {Z}_{\tau})\right] . \]

Assume we are given European put prices $\{ \tilde{P}_{T_i}(\tilde{k}) \}_{\tilde{k} \geq 0}$ for $i=1,2$ for a continuum of strikes $\tilde{k}$. If the call prices have come from a model for which the discounted price process is a martingale then
\[ \tilde{P}_{T_i}(\tilde{k}) = \frac{1}{\tilde{B}_{T_i}} \E[(\tilde{k} - \tilde{Z}_{T_i})^+] 
= \E\left[ \left( \frac{\tilde{k}}{\tilde{B}_{T_i}} - Z_i\right)^+ \right]=: P_i\left( \frac{\tilde{k}}{\tilde{B}_{T_i}} \right). \]
Then for fixed $i$ we have $P_i(k) = \tilde{P}_{T_i}(k \tilde{B}_{T_i})$, and if we are given European put prices with maturity $T_i$ then we can read off the law of $Z_{i}$:
\[ \Prob (Z_{i} < k) = P'_i(k-) = \frac{\partial}{\partial k} \tilde{P}_{T_i}(k \tilde{B}_{T_i} -) .\]

Henceforth we assume we work in a discounted setting and with time-index in the set $i=0,1,2$. In this setting the American put has payoff $(K_1 -Z_1)^+$ at time 1 and payoff $(K_2 - Z_2)^+$ at time 2 where $K_i = K/\tilde{B}_{T_i}$. Since interest rates are positive by hypothesis, we have $K_2<K_1$. We assume that we are given the prices of European puts (with maturities $T_1$ and $T_2$ in the original timescale) for all possible strikes. From these we can infer the laws of the discounted price process at times $1$ and $2$. We denote these laws by $\mu$ and $\nu$. It follows from Jensen's inequality that if $\mu$ and $\nu$ have arisen from sets of European put options in this way then $\mu \leq_{cx} \nu$.

\begin{defn}[Hobson and Neuberger~\cite{HobsonNeuberger:16}] Suppose $\mu \leq_{cx} \nu$.

Let $\sS = (\Omega, \sF, \Prob, \F = \{ \sF_0, \sF_1, \sF_2 \})$ be a filtered probability space. We say $M=(M_0,M_1,M_2)=(\bar{\mu},X,Y)$ is a $(\sS,\mu,\nu)$ consistent stochastic process and we write $M \in \sM(\sS, \mu, \nu)$ if
\begin{enumerate}
\item $M$ is a $\sS$-martingale,
\item $\sL(M_1) = \mu$ and $\sL(\sM_2) = \nu$.
\end{enumerate}
We say $(\sS,M)$ is a $(\mu,\nu)$-consistent model if $\sS$ is a filtered probability space and $M$ is a $(\sS,\mu,\nu)$ consistent stochastic process.
\end{defn}

Let $B \in \sF_1$. Define the stopping time $\tau_B$ by $\tau_B = 1$ on $B$ and $\tau_B=2$ on $B^c$. (Conversely, any stopping rule taking values in $\{1,2\}$ has a representation of this form.) Suppose $(\sS,M)$ is a $(\mu,\nu)$ consistent model. The $(\sS,M)$ model-based expected payoff of the American put under stopping rule $\tau_B$ is
\[ \sA(B, M, \sS) = \E[(K_{\tau_B} - M_{\tau_B})^+] . \]
Then, optimising over stopping rules under the model $(\sS,M)$ the price of the American put is
$ \sA(M,\sS) = \sup_B \sA(B,M,\sS) $.
The highest model based expected payoff for the American put is
\begin{equation}
\sP = \sup_{\sS} \sup_{M \in \sM(\sS, \mu, \nu)} \sup_B \sA(B,M,\sS) .
\label{eq:primal}
\end{equation}

\begin{rem}
It is important to note that the supremum in \eqref{eq:primal} can exceed the supremum in \eqref{eq:primalcts}, but only in the case where $\mu$ has atoms, see Hobson and Norgilas~\cite{HobsonNorgilas:18}. The supremum in \eqref{eq:primalcts} gives the highest model based price under the restriction that $\sF_0$ is trivial, $\sF_1 = \sigma(X)$ and $\sF_2 = \sigma(X,Y)$. However, as pointed out in Hobson and Neuberger~\cite{HobsonNeuberger:17}, see also Hobson and Neuberger~\cite{HobsonNeuberger:16}, Bayraktar and Zhou~\cite{BayraktarZhou:16} and Aksamit et al.~\cite{AksamitDengOblojTan:17}, it is sometimes possible to achieve a higher model price if we work on a richer probability space. In the financial context, the choice of probability space is typically not specified. Instead the choice of probability space is a modelling issue, and it seems unreasonable to restrict attention to a sub-class of models without good reason, especially if this sub-class does not include the optimum.

The case where $\mu$ has atoms will be excluded by our standing assumptions, so we find that it is always sufficient to work in a setting in which ${\mathbb F}$ is the natural filtration of $M$.
\label{rem:HN2}
\end{rem}

\subsection{Superhedging}
\label{ssec:superhedge}
The following notion of a robust superhedge for an American option was first introduced by Neuberger~\cite{Neuberger:07}, see also Bayraktar and Zhou~\cite{BayraktarZhou:16} and Hobson and Neuberger~\cite{HobsonNeuberger:17}.

We work in discounted units over two time-points.
Consider a general American-style option with payoff $a$ if exercised at time 1, and payoff $b$ if exercised at time 2, where $a:\R \mapsto \R_+$ and $b: \R \mapsto \R_+$ are positive functions.

\begin{defn}
$(\phi,\psi, \{\theta_i \}_{i = 1,2})$ is a superhedge for $(a,b)$ if
\begin{eqnarray}
\label{eq:a} a(x) & \leq & \phi(x) + \psi(y) + \theta_1(x)(y-x), \\
\label{eq:b} b(y) & \leq & \phi(x) +\psi(y) + \theta_2(x)(y-x).
\end{eqnarray}
The cost of a superhedge is given by
\[ \sC = \sC(\phi,\psi, \{\theta_i \}_{i = 1,2}; \mu,\nu) = \int \phi(x) \mu(dx) + \int \psi(y)\nu(dy) , \]
where we set $\sC = \infty$ if $\int \phi(x)^+ \mu(dx) + \int \psi(y)^+ \nu(dy) = \infty$.
We let ${\sH}(a,b)$ be the set of superhedging strategies $(\phi,\psi, \{\theta_i \}_{i = 1,2})$.
\end{defn}

The idea behind the definition is that the hedger purchases a portfolio of maturity-1 European puts (and calls) with payoff $\phi$ and a portfolio of maturity-2
European puts (and calls) with payoff $\psi$. (The fact that this can be done and has cost $\sC$ follows from arguments of Breeden and Litzenberger~\cite{BreedenLitzenberger:78}.)
In addition, if the American option is exercised at time 1 the hedger holds $\theta_1$ units of the underlying between times 1 and 2; otherwise the hedger holds
$\theta_2$ units of the underlying over this time-period. In the former case, \eqref{eq:a} implies that the strategy superhedges the American option payout; in the later case \eqref{eq:b} implies the same.

\begin{rem}
We could extend the definition and allow a holding of $\theta_0$ units of the discounted asset over the time-period $[0,1)$. Then the RHS of \eqref{eq:a} would be
\begin{equation}
\phi(x) + \psi(y) + \theta_0 (x-M_0) + \theta_1(x)(y-x).
\label{eq:addtime0}
\end{equation}
However, after a relabelling $\phi(x) + \theta_0(x-M_0) \mapsto \phi(x)$, \eqref{eq:addtime0} reduces to \eqref{eq:a}. (Note that $\int \theta_0(x-M_0) \mu(dx) =0$ by the martingale property so that $\sC$ is unchanged.) Similarly for \eqref{eq:b}. Hence there is no gain in generality by allowing non-zero strategies between times 0 and 1.
\end{rem}

The dual (superhedging) problem is to find
\begin{equation}
\label{eq:dual}
{\sD}(\mu, \nu; a,b) = \inf_{(\phi,\psi, \{\theta_i \}_{i = 1,2}) \in {\sH}(a,b)} \sC(\phi,\psi, \{\theta_i \}_{i = 1,2};\mu,\nu) .
\end{equation}

Potentially the space ${\sH}$ could be very large and it is extremely useful to be able to search over a smaller space. The next lemma shows that any convex $\psi$ with $\psi \geq b$ can be used to generate a superhedge $(\phi,\psi, \{\theta_i \}_{i = 1,2})$.

For a convex function $\chi$ let $\chi'_+$ denote the right-derivative of $\chi$.


\begin{lem}
Suppose $\psi \geq b$ with $\psi$ convex. Define $\phi = (a-\psi)^+$ and set $\theta_2=0$ and $\theta_1= - \psi'_+$. Then
$(\phi,\psi, \{\theta_i \}_{i = 1,2})$ is a superhedge.
\label{lem:phifrompsi}
\end{lem}

\begin{proof}
We have
\[ b(y) \leq \psi(y) \leq \phi(x)+\psi(y) = \phi(x) +\psi(y) + \theta_2(x)(y-x) \]
and \eqref{eq:b} follows. Also, by the convexity of $\psi$, $\psi(x) \leq \psi(y) - \psi'_+(x)(y-x)$ and
\[ a(x) \leq (a(x)-\psi(x))^+ + \psi(x) \leq \phi(x) + \psi(y) + \theta_1(x)(y-x). \]
Hence \eqref{eq:a} follows.
\end{proof}

Let $\breve{\sH}= \breve{\sH}(b)$ be the set of convex functions $\psi$ with $\psi \geq b$.
For $\psi \in \breve{\sH}$  we can define the associated cost of the portfolio
\[ \breve{\sC}(\psi; \mu,\nu) =  \int (a(x)-\psi(x))^+ \mu(dx) + \int \psi(y) \nu(dy). \]
The reduced dual hedging problem restricts attention to superhedges generated from $\psi \in \breve{\sH}$ and is to find
\begin{equation}
 \breve{\sD} = \breve{\sD}(\mu,\nu;a,b) = \inf_{\psi \in \breve{\sH}(b)} \breve{\sC}(\psi;a,b) .
 \label{eq:dualS}
\end{equation}
Clearly we have ${\sD} \leq \breve{\sD}$: we will show that ${\sD} = \breve{\sD}$ for the American put.


\subsection{Weak and Strong Duality}
\label{ssec:strong=weak}
Let $(\sS,M)$ be a $(\mu,\nu)$ consistent model and let $\tau$ be a stopping time in this framework. The expected payoff of the American put under this stopping rule is $\E[(K_\tau - M_\tau)^+]$. Conversely, let $\psi$ be any convex function with $\psi(y) \geq (K_2-y)^+$ and let $\phi(x)= [(K_1-x)^+ - \psi(x)]^+$ and $\theta_i(x) = - \psi'_+(x) I_{\{ i = 1 \}}$.
Then for any $i \in \{1,2\}$ we have $(K_i - M_i)^+ \leq \psi(M_2) + \phi(M_1) + \theta_i(M_1)(M_2-M_1)$
and hence for any random time $\tau$ taking values in $\{1,2\}$, $(K_\tau - M_\tau)^+ \leq \psi(M_2) + \phi(M_1) + \theta_\tau(M_1)(M_2-M_1)$.
Then $\E[(K_\tau - M_\tau)^+] \leq \E^{X \sim \mu, Y \sim \nu}[\phi(X) + \psi(Y)]$ and we have weak duality $\sP \leq \sD$.

Suppose we can find $(\sS^*, M^*, B^*)$ with $M^* \in \sM(\sS^*,\mu, \nu)$ and $\psi^* \in \breve{H}$ such that
\[ \sA(B^*,M^*, \sS^*) = \breve{\sC}(\psi^*, \mu, \nu) . \]
Then $\sA(B^*,M^*, \sS^*) \leq \sP \leq \sD \leq \breve{\sD} \leq  \breve{\sC}(\psi^*, \mu, \nu)$ but since the two outer terms are equal we have $\sP=\sD$ and strong duality. Moreover, $(\sS^*,M^*)$ is a consistent model which generates the highest price for the American put (and $\tau^*$ given by $\tau^*=1$ if and only if $X \in B^*$ is the optimal exercise rule) and $\psi^*$ generates the cheapest superhedge.

\subsection{The left-curtain coupling}
The left-curtain coupling (or martingale transport) was introduced by Beiglb\"{o}ck and Juillet~\cite{BeiglbockJuillet:16} and further studied by Henry-Labord\`{e}re and Touzi~\cite{HenryLabordereTouzi:16} and Beiglb\"{o}ck et al.~\cite{BeiglbockHenryLabordereTouzi:17}.

For real numbers $c,d$ with $c \leq x \leq d$ define the probability measure $\chi_{c,x,d}$ by $\chi_{c,x,d} = \frac{d-x}{d-c}\delta_c + \frac{x-c}{d-c} \delta_d$ with $\chi_{c,x,d} = \delta_x$ if $(d-x)(x-c)=0$. Note that $\chi_{c,x,d}$ has mean $x$. $\chi_{c,x,d}$ is the law of a Brownian motion started at $x$ evaluated on the first exit from $(c,d)$.

\begin{lem}[Beiglb\"ock and Juillet\cite{BeiglbockJuillet:16}, Corollary 1.6]\label{lem:LC}
Let $\mu,\nu$ be probability measures in convex order and assume that $\mu$ is continuous. Then there exists a pair of measurable functions $T_d : \R \mapsto \R$ and $T_u : \R \mapsto \R$ such that $T_d(x) \leq x \leq T_u(x)$, such that for all $x<x'$ we have $T_u(x) \leq T_u(x')$ and $T_d(x') \notin (T_d(x),T_u(x))$, and such that if
we define $\pi_{lc}(dx,dy) = \mu(dx) \chi_{T_d(x),x,T_u(x)}(dy)$ then $\pi_{lc} \in \hat{\Pi}_M(\mu,\nu)$. $\pi_{lc}$ is called the left-curtain martingale coupling.
\end{lem}

Note that there is no claim of uniqueness of the functions $T_d,T_u$ in Lemma~\ref{lem:LC}. For example, the definitions of $T_d$ and $T_u$ are immaterial outside $[\ell_\mu,r_\mu]$. Further, if $T_u$ has a (necessarily upward) jump at $x'$ then it does not matter what value we take for $T_u(x')$ provided $T_u(x') \in [T_u(x'-),T_u(x'+)]$. (Since we are assuming $\mu$ is continuous, the probability that we choose an $x$-coordinate value of $x'$ is zero.) More importantly, if $(T_d,T_u)$ satisfy the properties of Lemma~\ref{lem:LC} and if $T_u(x)= x$ on an interval $[\underline{x},\overline{x})$ then we can modify the definition of $T_d$ on $[\underline{x},\overline{x})$ to either $T_d(x)= x$ or $T_d(x)= T_d(\underline{x}-)$ and still satisfy the relevant monotonicity properties.
Henry-Labord\`{e}re and Touzi~\cite{HenryLabordereTouzi:16} resolve this indeterminacy by setting $T_d(x) = x$ on the set $T_u(x)=x$ and also taking $T_u$ and $T_d$ to be right-continuous.

We follow Henry-Labord\`{e}re and Touzi~\cite{HenryLabordereTouzi:16} by taking $T_d(x) = x$ on the set $T_u(x)=x$ but we do not make right-continuity assumptions on $T_d$ and $T_u$. Also we write $(f,g)$ in place of $(T_d,T_u)$.

\begin{lem}
Let $(T_d,T_u)$ be a pair of functions satisfying the monotonicity properties listed in Lemma~\ref{lem:LC}. Suppose they lead to a solution $\pi_{lc} \in \hat{\Pi}_M(\mu,\nu)$. \\
Set $g(x)=T_u(x)$. On $g(x)>x$ set $f(x)=T_d(x)$ and on $g(x)=x$ set $f(x)=x$. Then $(f,g)$
are such that $f(x) \leq x \leq g(x)$ and for all $x'>x$ we have $g(x') \geq g(x)$ and $f(x') \notin (f(x),g(x))$.
Moreover, $\mu(dx) \chi_{f(x),x,g(x)}(dy) = \mu(dx) \chi_{T_d(x),x,T_u(x)}(dy)$.
\end{lem}

\begin{proof} The property $f(x) \leq x \leq g(x)$ is immediate so we simply need to check that for $x'>x$ we have $g(x') \geq g(x)$ and $f(x') \notin (f(x),g(x))$. Monotonicity of $g$ is inherited from monotonicity of $T_u$. If $g(x)=x$ then $f(x)=x$ and $f(x') \notin (f(x),g(x))= \emptyset$. If $g(x)> x$ and $g(x')>x'$ then $f(x') = T_d(x') \notin (T_d(x),T_u(x)) = (f(x),g(x))$.  Finally, if $g(x)>x$ and $g(x') = x'$ then $f(x')=x' \notin (f(x),x'=g(x')) \supseteq (f(x),g(x))$.
\end{proof}

Figure~\ref{fig:generalFG} gives a stylized representation of $f$ and $g$ in the case where $\nu$ has no atoms. (Atoms of $\nu$ lead to horizontal sections of $f$ and $g$, see Section~\ref{ssec:atomsinnu}.) In the figure the set $\{ g(x)>x \}$ is a finite union of intervals whereas in general it may be a countable union of intervals. Similarly, in the figure $f$ has finitely many downward jumps, whereas in general it may have countably many jumps. Nonetheless Figure~\ref{fig:generalFG} captures the essential behaviour of $f$ and $g$. 

\begin{figure}[H]
\centering
\resizebox{8cm}{8cm}{
\begin{tikzpicture}[dot/.style={circle,inner sep=1pt, fill, label={#1},name#1},
extended line/.style={shorten >=-#1, shorten <=-#1},
extended line/.default=3cm,
declare function={	
diag(\x)=\x;
    	f=1;
	e1=1.5;
	g=3;
	e2=3.5;
	f2=4;
	f3=5;
	fx=5.5;
	e1x=6;
	gx=7;
	e2x=7.5;
	f2x=8;
	f3x=8.5;
	s=8.5;
	sx=8.5;
	k1=6.37;
	k2=4.5;
	a(\x)=(k1-\x)*(\x<k1)-8;
	b(\x)=(k2-\x)*(\x<k2)-8;
}]

           \draw[name path=diag,black] (0,0) -- (10,10);


          \coordinate (g) at (g, {diag(g)});
           \coordinate (e1) at (e1, {diag(e1)});
            \coordinate (e2) at (e2, {diag(e2)});
             \coordinate (s) at (s, {diag(s)});

             \coordinate (a) at (g, {diag(f)});
              \coordinate (b) at (f2, {diag(g)});
               \coordinate (c) at (f2, {diag(f)});

             \draw[thick, blue, name path=tu1] (e1) to[out=75,in=180] (g);
             \draw[thick, blue, name path=td1] (e1) to[out=315,in=180] (a);
              \draw[thick, blue, name path=tu2] (e2) to[out=75,in=205] (f3,{diag(f3)});
              \draw[thick, blue, name path=td2] (e2) to[out=300,in=140] (b);
              \draw[thick, blue, name path=td3,] (c) to[out=300,in=180] (f3,0.5);

               \coordinate (gx) at (gx, {diag(gx)});
           \coordinate (e1x) at (e1x, {diag(e1x)});
            \coordinate (e2x) at (e2x, {diag(e2x)});
             \coordinate (sx) at (sx, {diag(sx)});

             \coordinate (ax) at (gx, {diag(fx)});
              \coordinate (bx) at (f2x, {diag(gx)});
               \coordinate (cx) at (f2x, {diag(fx)});
               \draw[thick, blue, name path=tu1x] (e1x) to[out=75,in=180] (gx);
             \draw[thick, blue, name path=td1x] (e1x) to[out=315,in=180] (ax);
              \draw[thick, blue, name path=tu2x] (e2x) to[out=75,in=205] (10,10);
              \draw[thick, blue, name path=td2x] (e2x) to[out=300,in=140] (bx);
              \draw[thick, blue, name path=td3x] (cx) to[out=300,in=180] (f3x,f3);
               \draw[thick, blue, name path=td3x] (f3x,0.5) to[out=300,in=180] (10,0);

		 \draw [blue, dashed] (g) -- (b) -- (c) -- (a) -- (g);
		  \draw [blue, dashed] (gx) -- (bx) -- (cx) -- (ax) -- (gx);
		  \draw [blue, dashed] (f3,{diag(f3)}) -- (f3x,f3) -- (f3x,0.5) -- (f3,0.5) -- (f3,{diag(f3)});

                   \path[name path=base] (0,0) to (10,0);
                   
                   \node (Tu)[scale=1] at (8.5,9.4) {$g$};
                     \node (Tu)[scale=1] at (9.4,0.3) {$f$};

\end{tikzpicture}
}
\caption{Stylized plot of the functions $f$ and $g$ in the general case (with no atoms). Note that on the set $g(x)=x$ we have $f(x)=x$.}
\label{fig:generalFG}
\end{figure}
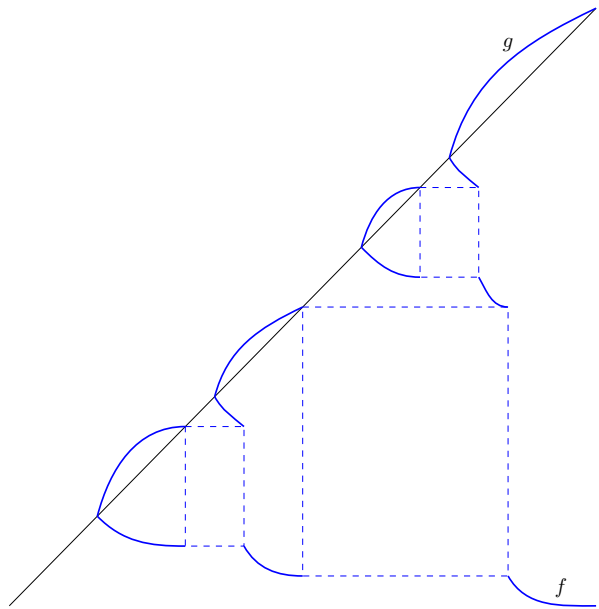

\begin{rem}
\label{rem:lcinterpretation}
The left-curtain martingale coupling can be identified with Figure~\ref{fig:generalFG} in the following way: choose an $x$-coordinate according to $\mu$; then if $g(x)=x$ set $Y=x=X$ so the pair $(X,Y)$ lies on the diagonal; otherwise if $g(x)>x$ then $f(x)<x$  and we set the $y$-coordinate to be $g(x)$ with probability $\frac{x-f(x)}{g(x)-f(x)}$ and $f(x)$ with probability $\frac{g(x)-x}{g(x)-f(x)}$. Then the coordinates $(x,y)$ represent the realised values of $(X,Y)$.

For a horizontal level $y$ there are two cases. Either, $g(y)>y$ and then the value of $y$ arises from a choice according to $\mu$ of $x=g^{-1}(y)$ for which $g(x)$ is chosen rather than $f(x)$; or $g(y)=y$  and the value $y$ arises either from a choice according to $\mu$ of $x=y$, or from a choice according to $\mu$ of $f^{-1}(y)$ combined with a choice of $y$-coordinate of $f(f^{-1}(y))=y$.
\end{rem}

Suppose $\nu$ is also continuous and fix $x$. Then, by the first paragraph of Remark~\ref{rem:lcinterpretation}, under the left-curtain martingale coupling mass in the interval $(f(x),x)$ at time 1 is mapped to the interval $(f(x), g(x))$ at time 2. Thus $\{ f(x), g(x) \}$ with $f(x) \leq x \leq g(x)$ are solutions to
\begin{eqnarray}
\int_{f}^x \mu(dz) & = & \int_{f}^g \nu(dz),  \label{eq:mass} \\
\int_{f}^x z \mu(dz) & = & \int_{f}^g z \nu(dz). \label{eq:mean}
\end{eqnarray}

Essentially, \eqref{eq:mass} is preservation of mass condition and \eqref{eq:mean} is preservation of mean and the martingale property. If $\nu$ has atoms then
\eqref{eq:mass} and \eqref{eq:mean} become
\begin{eqnarray}
\label{eq:massAtom}
\int_{f}^x \mu(dz) &  = & \int_{(f,g)} \nu(dz) + \lambda_f + \lambda_g, \\
\label{eq:meanAtom}
\int_{f}^x z \mu(dz) &  = & \int_{(f,g)} z \nu(dz) + f \lambda_f + g \lambda_g ,
\end{eqnarray}
respectively, where $0 \leq \lambda_f \leq \nu( \{ f \} )$ and $0 \leq \lambda_g \leq \nu( \{ g \} )$.

Returning to the case of continuous $\mu$ and $\nu$, for fixed $x$ there can be multiple solutions to \eqref{eq:mass} and \eqref{eq:mean}. If, however, we consider $f$ and $g$ as functions of $x$ and impose the additional monotonicity properties of Lemma~\ref{lem:LC} (for $x<x'$, $g(x) \leq g(x')$ and $f(x') \notin (f(x),g(x))$), then typically, for almost all $x$ there is a unique solution to \eqref{eq:mass} and \eqref{eq:mean}. However, there are exceptional $x$ at which $f$ jumps and at which there are multiple solutions, see Section~\ref{ssec:bimodal}.

\begin{rem}
\label{rem:trivial}
Note that if $g(x)=T_u(x)=x=f(x)>T_d(x)$ then we typically do not have $\int_{T_d(x)}^x \mu(dz) =  \int_{T_d(x)}^{g(x)} \nu(dz)$ and
$\int_{T_d(x)}^x z \mu(dz)  =  \int_{T_d(x)}^{g(x)} z \nu(dz)$. However, we trivially have $\int_{f(x)}^x \mu(dz) =  \int_{f(x)}^{g(x)} \nu(dz)$ and
$\int_{f(x)}^x z \mu(dz)  =  \int_{f(x)}^{g(x)} z \nu(dz)$. This explains our choice of $f(x)$ when $g(x)=x$.
\end{rem}

\begin{rem}
\label{rem:reverseconstruction}
There are many pairs $(\mu,\nu)$ which lead to the same pair of functions $(f,g)$. Conversely, let $\sI_1 \subseteq \sI_2 \subseteq \R$ be intervals and define
\[ \Xi^{\sI_1,\sI_2}  =  \{ (f,g) : g:\sI_1 \rightarrow \sI_2, g(x) \geq x, f:\sI_1 \rightarrow \sI_2, f(x)\leq x \}
\]
and let $\Xi = \cup_{\sI_1 \subseteq \sI_2} \Xi^{\sI_1,\sI_2}$.
Suppose $\mu$ is any integrable measure with support
in $\sI_1$ and define $\pi$ via $\pi(dx,dy) = \mu(dx) \chi_{f(x),x,g(x)}(dy)$ and $\nu$ via
\begin{equation}
\nu(dy) = \int_x \mu(dx) \chi_{f(x),x,g(x)}(dy).
\label{eq:nudef}
\end{equation}
Then (subject to integrability conditions\footnote{We require that
$\int \mu(dx) \frac{(g(x)-x)(x-f(x))}{g(x)-f(x)} I_{ \{ g(x)>f(x) \} } < \infty$. Then $\nu$ is integrable, $\mu \leq_{cx} \nu$ and $\pi$ is a martingale coupling.})
we have $\pi \in \hat{\Pi}_M(\mu,\nu)$.

Moreover, if we set
\[ \Xi^{\sI_1,\sI_2}_{Mon}  = \{ (f,g) \in \Xi^{\sI_1,\sI_2}: \mbox{$g$ increasing, $f(x)=x$ on $g(x)=x$, for $x'>x$ $f(x') \notin (f(x),g(x))$} \}  \]
and $\Xi_{Mon} = \cup_{I_1 \subseteq I_2} \Xi^{\sI_1,\sI_2}_{Mon}$, then provided the same integrability conditions are satisfied we have that if $\nu$ is given by \eqref{eq:nudef} then $\pi$ given by $\pi(dx,dy) = \mu(dx) \chi_{f(x),x,g(x)}(dy)$ is the left-curtain coupling.

The relevance of this remark is as follows. Given a pair $\mu \leq_{cx} \nu$ it may be difficult to determine the properties of $(f,g)$ which define the left-curtain coupling, beyond the fact that $(f,g) \in \Xi_{Mon}$. (For example, it may be difficult to ascertain the number of downward jumps of $f$ without calculating $f$ and $g$ everywhere.) However, if we want to construct examples for which $(f,g)$ have additional properties (such as no downward jump) then we can start with an appropriate pair $(f,g)$, take arbitrary (continuous) initial law $\mu$ with support on the interval where $f$ is defined, and then define $\nu$ via \eqref{eq:nudef}. This observation underpins our analysis in Sections~\ref{ssec:dispersion} and \ref{ssec:bimodal}.
\end{rem}

\section[Robust bounds when $\mu$ is atom-free]{Robust bounds for American puts when $\mu$ is atom-free}
\label{sec:continuous}

\subsection{Problem formulation}
\label{ssec:formulation}
Our goal in this section is to derive the highest consistent model price for the American put. We begin by giving a concise formulation of the problem, and stating a version of our main result. Then we first study the problem in a simple special case, second generalise to a case which exhibits all the main features and third present the analysis in the general case. 

Throughout this paper we assume that $\mu$ has no atoms. The same assumption is made in Beiglb\"{o}ck and Juillet~\cite{BeiglbockJuillet:16}, Henry-Labord\`{e}re and Touzi~\cite{HenryLabordereTouzi:16} and Beiglb\"{o}ck et al~\cite{BeiglbockHenryLabordereTouzi:17}. The extension of the left-curtain martingale coupling to the case where $\mu$ has atoms is the subject of Hobson and Norgilas~\cite{HobsonNorgilas:18}.

\begin{sass}
$\mu$ has no atoms.
\label{sa:noatoms}
\end{sass}

We consider an American put on an asset. Under the bond numeraire, we represent the price of the underlying security by $M = (M_0 = \bar{\mu}, M_1 = X, M_2=Y)$. The American put may only be exercised at time $1$ or time $2$: if the put is exercised at time 1 the payoff is $(K_1 - X)^+$; if the put is exercised at time 2   the payoff is $(K_2-Y)^+$. We say the put is in-the-money at time 1 (respectively time 2) if $X < K_1$ (respectively $Y<K_2$). Otherwise the put is out-of-the-money. The laws of $X$ and $Y$ are presumed to be given and $\sL(X) = \mu$ and $\sL(Y)= \nu$. 

Under Standing Assumption~\ref{sa:noatoms} our problem is to
\begin{prob}
Find
\begin{enumerate}
\item the highest possible expected payoff of the American option, where expectations are calculated under models which are consistent with the marginal laws of $M$,
\item the cheapest superhedging price.
\end{enumerate}
\end{prob}
Our main result is as follows:
\begin{thm}
\label{thm:main}
The highest model-based expected payoff of the American put is equal to the cheapest superhedging price. Moreover, the highest model-based expected payoff is attained by the model associated with the left-curtain martingale coupling (and a judiciously chosen stopping rule). Further, we can characterise the cheapest super-hedging strategy: it takes the form described in Lemma~\ref{lem:phifrompsi} and it is one of four possible types.
\end{thm}

We begin by considering a couple of degenerate cases.

If $K_1 \leq \ell_\mu$ then the American put is always out-of-the-money at time 1, and the American put is equivalent to the European put with strike $K_2$ and maturity $2$. Since puts with strike $K_1$ and maturity 1 are costless, a simple superhedging strategy is to purchase one European put with strike $K_1$ and maturity 1, and one European put with strike $K_2$ and maturity $2$. The cost of this hedge is $P_\nu(K_2)$, this is also the model-based expected payoff of the American put under {\em any} consistent model.

If $K_1 \leq K_2$ then $\E[(K_2-Y)^+|X] \geq (K_2 - X)^+ \geq (K_1-X)^+$
and $\tau=2$ is optimal. Again the American put is equivalent to the European put with strike $K_2$ and maturity $2$. In this case, for a superhedge it is sufficient to purchase one European put with strike $K_2$ and maturity $2$. By Lemma~\ref{lem:phifrompsi} (with $\psi(y) = (K_2-y)^+$ and $\phi=0$) this generates a superhedge with cost $P_\nu(K_2)$. Again, this is the the model-based expected payoff of the American put under {\em any} consistent model.

For the remainder of the paper we make
\begin{sass}
$K_1 > \max \{ \ell_\mu , K_2 \}$.
\end{sass}

\subsection[Dispersion assumption]{American puts under the dispersion assumption}
\label{ssec:dispersion}

\subsubsection{The left-curtain coupling}
The goal in this section is to present the theory in a simple special case, and to illustrate the main features and solution techniques of our approach unencumbered by technical issues or the consideration of exceptional cases.
The following assumption is a small modification of one introduced by Hobson and Klimmek~\cite{HobsonKlimmek:15}, see also Henry-Labord\`ere and Touzi~\cite{HenryLabordereTouzi:16}. See Figure~\ref{fig:densities}.

\begin{ass}[Dispersion Assumption]
$\mu$ and $\nu$ are absolutely continuous with continuous densities $\rho$ and $\eta$, respectively. $\nu$ has support on $(\ell_\nu,r_\nu) \subseteq (-\infty,\infty)$ and $\eta>0$ on $(\ell_\nu,r_\nu)$. $\mu$ has support on $(\ell_\mu,r_\mu) \subseteq (\ell_\nu,r_\nu)$ and $\rho>0$ on $(\ell_\mu,r_\mu)$. In addition: \\
$(\mu-\nu)^+$ is concentrated on an interval $E = (e_-,e_+)$ and $\rho>\eta$ on $E$; \\
$(\nu - \mu)^+$ is concentrated on $(\ell_\nu, r_\nu) \setminus E$ and $\eta>\rho$ on $(\ell_\nu,e_-) \cup (e_+,r_\nu)$.
\label{ass:dispersion}
\end{ass}

If $\mu \leq_{cx} \nu$ are centred normal distributions with different variances or distinct lognormal random variables with common mean then Assumption~\ref{ass:dispersion} is satisfied.

Under the Dispersion Assumption $\{ k:D_{\mu,\nu}(k)>0 \}$ is an interval and $D=D_{\mu,\nu}$ is convex to the left of $e_-$, concave on $(e_-,e_+)$ and again convex above $e_+$. 

\begin{figure}[H]
\centering
\begin{tikzpicture}[
declare function={	
    	phi(\x,\m,\s)=(1/sqrt(2*pi*\s))*exp(-pow(\m-\x,2)/(2*\s));
	s1=0.7;
	s2=2;
	m=0;
	f=-2;
	z=0.15;
	g=0.8;
	e1=sqrt(-(ln(sqrt(s2))-ln(sqrt(s1)))/(pow(2*s2,-1)-pow(2*s1,-1)));
	}]
\begin{axis}[axis lines=middle,
            ytick=\empty,
            xtick=\empty,
             xmin=-4, xmax=4,
             ymin=-0.15, ymax=0.5,
             yticklabels={},
            xticklabels={},
            axis line style={draw=none}]

\addplot[name path=base,black,domain={-5:5}] {0} node[pos=1, below]{};
\addplot[name path=rho,black,domain={-5:5},samples=100] {phi(\x,m,s1)} node at (0.55,0.45) {$\rho$};
\addplot[name path=eta,black,domain={-5:5},samples=100] {phi(\x,m,s2)} node at (1.7,0.17) {$\eta$};

\path [name path=lineA](f,0)--(f,0.5);
\draw [name intersections={of=lineA and eta}, black] (f,0) -- (intersection-1);
\draw[dashed] (f,0) -- (f,-0.06) node[below]  {$f$};

\path [name path=lineB](-e1,0)--(-e1,0.5);
\draw [name intersections={of=lineB and eta}, black, dashed] (-e1,0) -- (intersection-1) node at (-e1+0.1,-0.03) {$e_-$};

\path [name path=lineC](z,0)--(z,0.5);
\draw [name intersections={of=lineC and rho}, black] (z,0) -- (intersection-1);
\draw[dashed] (z,0) -- (z,-0.06) node[below]  {$x$};

\path [name path=lineD](g,0)--(g,0.5);
\draw [name intersections={of=lineD and eta}, black] (g,0) -- (intersection-1);
\draw[dashed] (g,0) -- (g,-0.06) node[below]  {$g$};

\path [name path=lineE](e1,0)--(e1,0.5);
\draw [name intersections={of=lineE and eta}, black,dashed] (e1,0) -- (intersection-1) node at (e1+0.1,-0.03) {$e_+$};

\addplot[pattern=crosshatch, pattern color=red!50] fill between[of=eta and rho, soft clip={domain=f:-e1}];
\addplot[pattern=crosshatch, pattern color=red!50] fill between[of=eta and rho, soft clip={domain=-e1:z}];
\addplot[pattern=crosshatch, pattern color=red!50] fill between[of=eta and base, soft clip={domain=z:g}];
\addplot[pattern=north west lines, pattern color=blue!50] fill between[of=rho and base, soft clip={domain=f:-e1}];
\addplot[pattern=north west lines, pattern color=blue!50] fill between[of=eta and base, soft clip={domain=-e1:z}];

	\draw[-latex', thick] (-0.2, .35) to[out=0,in=90] (0.5, 0.15);
	\draw[-latex', thick] (-0.3, .35) to[out=180,in=90] (-1.7, 0.1);

\end{axis}
\end{tikzpicture}

\caption{Sketch of the densities $\rho$ and $\eta$ and the locations of $f,g$ for given $x>e_-$. Time-1 mass in the interval $(f,x)$ stays in the same place if possible. Mass which cannot stay constant is mapped to $(f,e_-)$ or $(x,g)$ in a way which respects the martingale property.}
\label{fig:densities}
\end{figure}
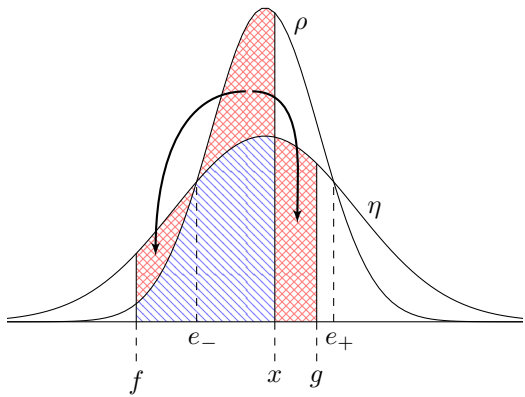



\begin{lem}[Henry-Labord\`{e}re and Touzi~\cite{HenryLabordereTouzi:16}, Section 3.4]\label{lem:LCdispersion}
Suppose Assumption~\ref{ass:dispersion} holds.
For all $x \in (e_-, r_\mu)$, there exist $f,g$ with $f < e_- < x < g$ such that
\eqref{eq:mass} and \eqref{eq:mean} hold. Moreover, if we consider $f$ and $g$ as functions of $x$ on $(e_-,r_\mu)$ then $f$ and $g$ are continuous, $f$ is strictly decreasing and $g$ is strictly increasing, $\lim_{x \downarrow e_-} f(x) = e_- = \lim_{x \downarrow e_-} g(x)$, $\lim_{x \uparrow r_\mu} f(x) = \ell_\nu$ and $\lim_{x \uparrow r_\mu} g(x) = r_\nu$. Finally, if we extend the domain of $f$ and $g$ to $[\ell_\mu,r_\mu]$ by setting $f(x)=x=g(x)$ on $[\ell_\mu, e_-]$ and $f(r_\mu) = \ell_\nu$ and $g(r_\mu)=r_\nu$ then $(f,g) \in \Xi^{[\ell_\mu,r_\mu], [\ell_\nu,r_\nu]}_{Mon}$.
\label{lem:FGconstruction}
\end{lem}
\begin{figure}[H]
\centering
\begin{tikzpicture}[scale=1]
\begin{axis}[axis lines=middle,
            ytick=\empty,
            xtick=\empty,
             xmin=0,xmax=20,
             ymin=-1,ymax=21,
             yticklabels={},
            xticklabels={},
            axis line style={draw=none}]

\addplot[name path=diag,black,domain={0:20}] {x} node[pos=1, below]{};

\draw[name path=tu1,thick] (4, 4) to[out=85,in=225] (15.5,17.5);
\draw[name path=tu2, thick] (15.5, 17.5) to[out=45,in=225] (20,20.5 );
\draw[name path=td, thick] (4, 4) to[out=335,in=180] (20, 0);

\addplot[pattern=crosshatch, pattern color=red!50]fill between[of=diag and td, smooth clip={domain=2:10}];
\path [name path=hor](0,0) to (20,0);
\addplot[pattern=north west lines, pattern color=black!50]fill between[of=td and hor, smooth clip={domain=0:1}];

\node (e)[scale=0.7] at (6,4) {$(e_-,e_-)$};
\node (f)[scale=1] at (16,1) {$f$};
\node (g)[scale=1] at (16,19) {$g$};

\end{axis}
\end{tikzpicture}

\caption{Sketch of functions $f$ and $g$ under the Dispersion Assumption, with the regions $K_2 < f(K_1)$ and $K_2 > f(K_1)$ shaded. This is a simple special case of Figure~\ref{fig:generalFG}.}
\label{fig:fgdispersion}
\end{figure}
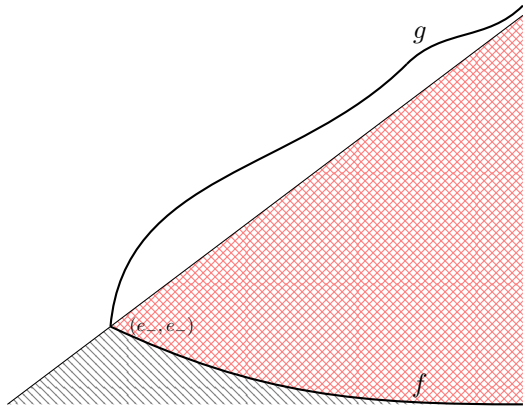

\begin{rem}
\label{rem:pilcsimple}
As discussed at the end of Remark~\ref{rem:reverseconstruction}, for the purposes of the analysis of this section it is not the fact that the measures $\mu$ and $\nu$ satisfy the Dispersion Assumption which is important, but rather that $\pi_{lc}$ is so simple, and $\{k : g(k)>k \}$ is a single interval on which $f$ is a monotone decreasing function.

Starting with monotonic $f$ and $g$, letting $\mu$ be continuous and defining $\nu$ by $\nu(dy) = \int_x \mu(dx) \chi_{f(x),x,g(x)}(dy)$ and $\pi_{lc}$ by
\begin{equation}
\label{eq:pilcdef}
 \pi_{lc}(dx,dy) = \mu(dx) \delta_x(dy) I_{\{ x \leq e_- \}} + \mu(dx) \chi_{f(x),x,g(x)} (dy) I_{\{ x > e_- \}},
\end{equation}
the pair $(\mu,\nu)$ may or may not satisfy Assumption~\ref{ass:dispersion} but nonetheless, a candidate optimal model, stopping time and hedge can be constructed exactly as described in this section, and can be proved to be optimal by the methods of this section.

Since our analysis depends on the pair $(\mu,\nu)$ only through the functions $(f,g)$ we may take as our starting point any $(f,g) \in \Xi_{Mon}$.
\end{rem}

\begin{rem}
\label{rem:ode}
In a related problem, Hobson and Klimmek~\cite{HobsonKlimmek:15} show how under the Dispersion Assumption, upper and lower functions can be characterised as solutions of a pair of coupled differential equations.
In our case $(f,g)$ solve a pair of coupled differential equations on $[e_-,r_\mu)$ obtained from differentiating \eqref{eq:mass} and \eqref{eq:mean}:
\begin{align*}
\frac{df}{dx} \phantom{-}&= - \frac{g-x}{g-f} \frac{\rho(x)}{\eta(f)-\rho(f)}, \\
\frac{dg}{dx} \phantom{-}&= \phantom{-}  \frac{x-f}{g-f} \frac{\rho(x)}{\eta(g)},
\end{align*}
with the initial condition $f(e_-)= e_- =g(e_-)$. See also Henry-Labord\`{e}re and Touzi~\cite[Equations (3.10) and (3.9)]{HenryLabordereTouzi:16}.
\end{rem}

The principle behind the left-curtain martingale coupling in Beiglb\"{o}ck-Juillet~\cite{BeiglbockJuillet:16} is that they determine where to map mass at $x$ at time 1 sequentially working from left to right. In our current setting there is an interval $(\ell_\mu,e_-]$ on which mass can remain unmoved between times 1 and 2.
To the right of $e_-$ we can define $f,g$ in such a way that mass is moved as little as possible. This leads to the ODEs in Remark~\ref{rem:ode}.

\subsubsection{The American put}
\label{ssec:amput}

Suppose $K_1 \in (e_-,r_\mu]$ and suppose $f$ and $g$ are constructed as in Lemma~\ref{lem:FGconstruction}.
Define $\Lambda:[g^{-1}(K_1),K_1] \mapsto \R$ by
\begin{equation} \Lambda(x) = \frac{(K_2 - f(x))-(K_1-x)}{x-f(x)} - \frac{(K_1 -x)}{g(x)-x} = \frac{(g(x)-K_1)}{g(x)-x} - \frac{(K_1 -K_2)}{x-f(x)}.
\label{eq:Lambdadef}
\end{equation}
Pictorially $\Lambda$ is the difference in slope of the two dashed lines in Figure~\ref{fig:xfgLambda}.

\begin{figure}[H]
\centering
\begin{tikzpicture}[scale=0.5,
    declare function={	
    	k1=10;
	k2=5;
	f=2;
	f1=0.5;
	g=13;
	g1=14.5;
	z=8;
	z1=9;
	a(\x)=(k1-\x)*(\x<k1);
	b(\x)=(k2-\x)*(\x<k2);
		}]

	\draw[name path=a, black, thick] plot[domain=0:15, samples=100] (\x,{a(\x)});
	\draw[name path=b, black, thick] plot[domain=0:15, samples=100] (\x,{b(\x)}) ;
	
	\draw[name path=h1, blue, thick, densely dashed] plot[domain=f:z, samples=100] (\x,{b(f)-((b(f)-a(z))/(z-f))*(\x-f)}) ;
	\draw[name path=h2, blue, thick, densely dashed] plot[domain=z:g, samples=100] (\x,{a(z)-(a(z)/(g-z))*(\x-z)});

\node (strike1) at (k1,-0.5) {$K_1$};
\node (strike2) at (k2,-0.5) {$K_2$};
\node[blue] (f) at (f,-0.5) {$f$};
\node[blue] (z) at (z,-0.5) {$x$};
\node[blue] (g) at (g,-0.5) {$g$};

\node (slope1)[blue] at (10,5) {\text{slope } $\frac{(K_2-f)-(K_1-x)}{x-f}$};
\node (slope2)[blue] at (15,3) {\text{slope } $\frac{(K_1-x)}{g-x}$};

	\draw[-latex', blue] (12.5, 3) to[out=180,in=45] (10, 1.5);
	\draw[-latex', blue] (6.5, 5) to[out=180,in=75] (2.5,3);


\draw [black, dashed] (f,0) -- (f,{b(f)});

\draw [black,  dashed] (z,0) -- (z,{a(z)});

\draw [black,  dashed] (g,0) -- (g,{a(g)});

\end{tikzpicture}
\caption{Sketch of put payoffs with points $x$, $f$ and $g$ marked. $\Lambda(x)$ is the difference in slope of the two dashed lines. }
\label{fig:xfgLambda}
\end{figure}
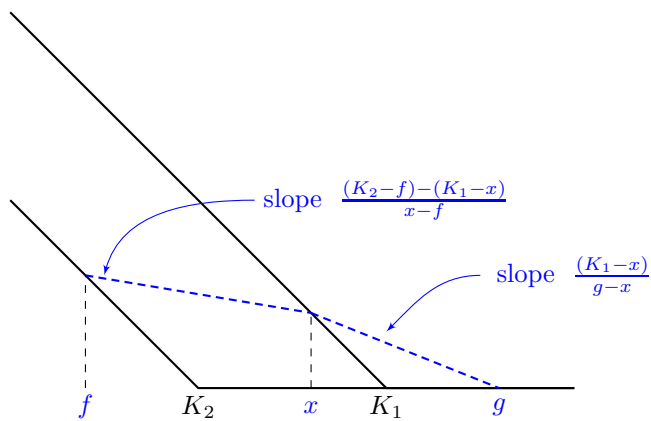

\begin{lem}
Suppose $K_1 \in (e_-,r_\mu]$ and $f(K_1) < K_2$.
Then there is a unique $x^*=x^*(\mu,\nu;K_1,K_2) \in (g^{-1}(K_1),K_1)$ such that $\Lambda(x^*)=0$. Moreover $f(x^*)<K_2$ and
\begin{equation}
\label{eq:propertiesx*}
\frac{(K_2 - f(x^*))}{g(x^*)-f(x^*)} = \frac{(K_1 -x^*)}{g(x^*)-x^*} = \frac{((x^* - f(x^*) -(K_1 -K_2))}{x^* - f(x^*)} = 1 - \frac{(K_1 -K_2))}{x^* - f(x^*)}.
\end{equation}
\label{lem:uniquex*}
\end{lem}

\begin{proof}
It is clear, see Figure~\ref{fig:xfgLambda}, that since $f$ and $g$ are continuous monotonic functions we have that $\Lambda$ is continuous and strictly increasing. Moreover, $\Lambda(g^{-1}(K_1))= -\frac{(K_1-K_2)}{g^{-1}(K_1) - (f \circ g^{-1})(K_1)} < 0$ and $\Lambda(K_1)=\frac{K_2-f(K_1)}{K_1 - f(K_1)}>0$ by hypothesis. Hence there is a unique root to $\Lambda=0$. At this root the equalities in \eqref{eq:propertiesx*} hold.
\end{proof}

Suppose $K_1 > e_-$ and $f(K_1) < K_2$ and that $x^*=x^*(\mu,\nu;K_1,K_2) \in (e_-,K_1)$ is such that $\Lambda(x^*)=0$.
We define a martingale coupling as follows (the superscript $*$ denotes the fact that the quantities we define are candidates for optimality):

Consider $\Omega^* = \R \times \R = \{\omega=(\omega_1,\omega_2) \}$ equipped with the Borel-sigma algebra $\sF^* = \sB(\Omega^*)$. Let $(X,Y)$ be the co-ordinate process so that $X(\omega)=\omega_1$ and $Y(\omega)=\omega_2$.
Let $\sF^*_0$ be trivial and suppose $\sF^*_1 = \sigma(X)$ and $\sF^*_2= \sigma(X,Y)$. Finally, for $\hat{\pi} \in \Pi_M(\mu,\nu)$ let $\Prob_{\hat{\pi}}$ be the martingale coupling $\Prob_{\hat{\pi}}(X\in dx, Y \in dy) = \hat{\pi}(dx,dy)$ and set
$\sS^* = (\Omega^*, \sF^*, \F^* := (\sF^*_0, \sF^*_1, \sF^*_2), \Prob_{\hat{\pi}})$. Then $(\sS^*, M_{\hat{\pi}} = (\bar{\mu},X,Y))$ is a $(\mu,\nu)$-consistent model.

It is easy to find a martingale coupling of $\mu$ and $\nu$ such that $(f(x^*),x^*)$ is mapped to $(f(x^*),g(x^*))$.
For example, we may take $\hat{\pi} = \pi_{lc} = \pi_{lc}(\mu,\nu)$, the left-curtain martingale coupling of Beiglb\"{ock} and Juillet~\cite{BeiglbockJuillet:16}.  More generally
let $\hat{\pi}_{x^*} \in \hat{\Pi}_M(\mu,\nu)$ be any martingale coupling such that $\hat{\pi}_{x^*}$ maps $(f(x^*),x^*)$ to $(f(x^*),g(x^*))$ and $(f(x^*),x^*)^C$ to $(f(x^*),g(x^*))^C$. The martingale coupling represented in Figure~\ref{fig:fgdispersion} has this property.
Let $M^*=(M^*_0=\bar{\mu},M^*_1=X,M^*_2=Y)$ be the stochastic process such that $\Prob(X \in dx, Y \in dy) = \hat{\pi}_{x^*}(dx,dy)$.
Then $M^* \in \sM(\sS^*, \mu, \nu)$.
Let $\tau^*$ be the stopping time such that $\tau^*=1$ on $(-\infty,x^*)$ and $\tau^*=2$ otherwise.
Our claim in Theorem~\ref{thm:dispersion} below is that $(\sS^*,M^*)$ and the stopping time $\tau^*$ are such that the model-based price of the American put under this stopping time is the highest possible, over all consistent models.

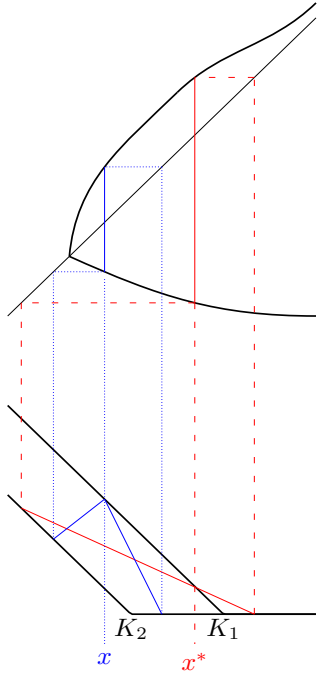
\begin{figure}[H]
\centering
\resizebox{4.5cm}{9cm}{
\begin{tikzpicture}[scale=0.5,
    declare function={	
    	k1=7;
	k2=4;
	a(\x)=(k1-\x)*(\x<k1);
	b(\x)=(k2-\x)*(\x<k2);
		}]

	\draw[name path=diag] (0,10) -- (10,20) node[right] {};
	
	\draw[name path=a, black, thick] plot[domain=0:10, samples=100] (\x,{a(\x)}) node[right] at (3,0.15) {};
	\draw[name path=b, black, thick] plot[domain=0:10, samples=100] (\x,{b(\x)}) node[right] at (3,0.15) {};
	
	\draw[name path=tu1, thick] (2, 12) to[out=85,in=225] (5.5,17.5);
	\draw[name path=tu2, thick] (5.5, 17.5) to[out=45,in=225] (10,20.5 );
	\draw[name path=td, thick] (2, 12) to[out=335,in=180] (10, 10);

\path [name path=lineA](8,0)--(8,20);
\draw [name intersections={of=lineA and diag}, red, loosely dashed] (8,0) -- (intersection-1);

\gettikzxy{(intersection-1)}{\ax}{\ay};
\path [name path=lineB](\ax,\ay)--(0,\ay);
\draw [name intersections={of=lineB and tu2}, red, loosely dashed] (\ax,\ay) -- (intersection-1);

\gettikzxy{(intersection-1)}{\ax}{\ay};
\path [name path=lineC](\ax,\ay)--(\ax,0);
\draw [name intersections={of=lineC and td}, red] (\ax,\ay) -- (intersection-1);

\gettikzxy{(intersection-1)}{\ax}{\ay};
\draw [red, loosely dashed] (\ax,\ay) -- (\ax,-1) node[scale=1.2, below,red] {$x^*$};
\path [name path=lineD](\ax,\ay)--(0,\ay);
\draw [name intersections={of=lineD and diag}, red, loosely dashed] (\ax,\ay) -- (intersection-1);

\gettikzxy{(intersection-1)}{\ax}{\ay};
\path [name path=lineE](\ax,\ay)--(\ax,0);
\draw [name intersections={of=lineE and b}, red, loosely dashed] (\ax,\ay) -- (intersection-1);
\draw[red, -] (intersection-1) -- (8,0) node[right] {};

\path [name path=lineA](5,0)--(5,20);
\draw [name intersections={of=lineA and diag}, blue, densely dotted] (5,0) -- (intersection-1);

\gettikzxy{(intersection-1)}{\ax}{\ay};
\path [name path=lineB](\ax,\ay)--(0,\ay);
\draw [name intersections={of=lineB and tu1}, blue, densely dotted] (\ax,\ay) -- (intersection-1);

\gettikzxy{(intersection-1)}{\ax}{\ay};
\path [name path=lineC](\ax,\ay)--(\ax,0);
\draw [name intersections={of=lineC and td}, blue] (\ax,\ay) -- (intersection-1);

\gettikzxy{(intersection-1)}{\ax}{\ay};
\node (temp) at (\ax,\ay) {};

\draw [name intersections={of=lineC and a}, blue,densely dotted] (temp) -- (intersection-1);

\gettikzxy{(intersection-1)}{\ax}{\ay};
\draw  [blue,densely dotted]   (intersection-1)--(\ax,-1) node[scale=1.2, below, blue] {$x$};
\node (temp1) at (\ax,\ay) {};

\gettikzxy{(temp)}{\ax}{\ay};
\path [name path=lineD](temp)--(0,\ay);
\draw [name intersections={of=lineD and diag}, blue,densely dotted] (\ax,\ay) -- (intersection-1);

\gettikzxy{(intersection-1)}{\ax}{\ay};
\path [name path=lineE](\ax,\ay)--(\ax,0);
\draw [name intersections={of=lineE and b}, blue,densely dotted] (\ax,\ay) -- (intersection-1);

\gettikzxy{(temp1)}{\ax}{\ay};
\draw[blue, -] (intersection-1) -- (\ax,\ay) node[right] {};
\draw[blue, -] (\ax,\ay) -- (5,0) node[right] {};

\node [scale=1.2] (strike1) at (k1,-0.5) {$K_1$};
\node [scale=1.2] (strike2) at (k2,-0.5) {$K_2$};

\end{tikzpicture}
}
\caption{A combination of Figures~\ref{fig:fgdispersion} and \ref{fig:xfgLambda}, showing how jointly they define the best model and best hedge. By adjusting $x$ we can find $x^*$ such that $\Lambda(x^*)=0$. Together the quantities $(f(x^*),x^*,g(x^*))$ define the optimal model, stopping time and hedge. }
\label{fig:Lambda}
\end{figure}

Continue to suppose $K_1 > e_-$ and $f(K_1) < K_2$.
Now we define a superhedge of the American put.
Let $\psi^*$ be the function
\begin{equation}
\psi^*(z) = \left\{ \begin{array}{ll} (K_2 - z)    &   z \leq f(x^*); \\
                                    \frac{(g(x^*) - z)(K_2 - f(x^*))}{g(x^*) - f(x^*)} \hspace{5mm} & f(x^*) < z \leq g(x^*); \\
                                    0 & z > g(x^*).
                                    \end{array} \right.
\label{eq:psi*}
\end{equation}
Note that by construction and by \eqref{eq:propertiesx*}, $\frac{K_2-f(x^*)}{g(x^*)-f(x^*)} = \frac{K_1-x^*}{g(x^*)-x^*}$ and then $\psi^*(x^*) = K_1-x^*$. Moreover, $\psi^*$ is convex and satisfies $\psi^*(z) \geq (K_2-z)^+$. Hence by Lemma~\ref{lem:phifrompsi}, $\psi^*$ can be used to construct a superhedge $(\psi^*, \phi^*,\theta_{1,2}^*)$.

In the following theorem we will assume the American put is not always strictly in-the-money at time 1 (or equivalently, $K_1 \leq r_\mu$). Discussion of the case $K_1>r_\mu$ is postponed until Section~\ref{sssec:K1>rmu} below.

\begin{thm} Suppose Assumption~\ref{ass:dispersion} holds.
\begin{enumerate}
\item
Suppose $K_1 \in (e_-, r_\mu]$ and $f(K_1) < K_2$. The model $(\sS^*,M^*)$  described in the previous paragraphs is a consistent model for which the price of the American option is the highest. The stopping time $\tau^*$ is the optimal exercise time. The function $\psi^*$ defined in \eqref{eq:psi*} defines the cheapest superhedge. Moreover, the highest model-based price is equal to the cost of the cheapest superhedge.
\item Suppose
either that Case A: $K_1 \leq e_-$
or that Case B: $K_1 \in (e_-, r_\mu]$ and $f(K_1) \geq K_2$.
Then there is a consistent model for which $(Y<K_2) = (X<K_2) \cup (X>K_1, Y< K_2)$ and any model with this property with the stopping rule $\tau=1$ if $X<K_1$ and $\tau=2$ otherwise attains the highest consistent model price. The cheapest superhedge is generated from ${\psi}(x) = (K_2-x)^+$ and the highest model-based price is equal to the cost of the cheapest hedge.
\end{enumerate}
\label{thm:dispersion}
\end{thm}

\begin{rem}
In Part 2 of the Theorem~\ref{thm:dispersion}, the left-curtain coupling generates a model which, when associated with the stopping rule of the theorem, attains the highest consistent model price.
\end{rem}

\begin{proof}

1. Suppose $K_1 > e_-$ and $f(K_1) < K_2$.
Then by Lemma~\ref{lem:uniquex*} there is a unique $x^* \in (g^{-1}(K_1),K_1)$ such that $\Lambda(x^*)=0$.  For this $x^*$ we can find $f^*=f(x^*)$ and $g^*=g(x^*)$ with $f^*<K_2$ and $K_1 < g^*$ such that $\frac{K_2 -f^*}{g^*-f^*}= \frac{K_1 - x^*}{g^* - x^*}$. For typographical reasons we abbreviate this $(x^*,f^*,g^*)$ to $(x,f,g)$ for the rest of this proof.

Since $\nu$ is continuous we have that $f,x,g$ solve \eqref{eq:mass} and \eqref{eq:mean}. The elements $f,x,g$ can be used to define a model using the construction after Lemma~\ref{lem:uniquex*} above. For this model we can calculate the expected payoff of the American put. At the same time we can use $(f,x,g)$ to define a superhedge. The remaining task is to show that the cost of the superhedge equals that of the model-based expected payoff. Then by the discussion in Section~\ref{ssec:strong=weak} we have found an optimal model and a cheapest superhedge.

The expected payoff of the American put (for this model and stopping rule) is
\begin{eqnarray*}
\lefteqn{ \int_{-\infty}^{x} (K_1-w) \mu(dw) + \int_{-\infty}^f (K_2 - w) (\nu - \mu)(dw) } \\
& = & P_\mu(x) + (K_1 - x) P'_{\mu}(x) + D(f) + (K_2 - f) D'(f).
\end{eqnarray*}

Now we consider the hedging cost. Set $\Theta = \frac{K_2 - f}{g-f} \in (0,1)$. Note that, since $x$ is such that $\Lambda(x)=0$ we have $\Theta = \frac{K_1-x}{g-x}$.
Recall the definition of $\psi^*$ in \eqref{eq:psi*}. Then
\[ \psi^*(y) = \Theta (g-y)^+ + (1- \Theta)(f-y)^+ . \]
Using Lemma~\ref{lem:phifrompsi} we can use $\psi^*$ to generate a superhedging strategy. The cost of this strategy is
\begin{equation}
\label{eq:price}
 \Theta P_\nu(g) + (1- \Theta)P_\nu(f) + (1- \Theta)\left[ P_\mu(x) - P_\mu(f) \right],
\end{equation}
where the first two terms arise from the purchase of the static time-2 portfolio $\psi^*$ and the third comes from the purchase of the time-1 portfolio $(K_1-w)^+ - \psi^*(w)$.  The expression in \eqref{eq:price} can be rewritten as
\[ P_\mu(x) + D(f) + \Theta \left[ P_\nu(g) - P_\nu(f) - P_\mu(x) + P_\mu(f) \right] . \]

Now we consider the difference between the hedging cost ($HC$) and the model-based expected payoff ($MBEP$).
Recall that $P_\chi(k) = \int_{-\infty}^k (k - x) \chi(dx)$, $\chi\in\{\mu,\nu\}$,  and that $D(k)=D_{\mu,\nu}(k) = P_{\nu}(k) - P_{\mu}(k)$. Then \eqref{eq:mass} and \eqref{eq:mean} can be rewritten as
\begin{eqnarray}
P'_{\mu}(x) - P'_{\mu}(f) & = & P'_{\nu}(g) - P'_{\nu}(f), \label{eq:massP} \\
(xP'_{\mu}(x) - P_{\mu}(x)) - (fP'_{\mu}(f) - P_{\mu}(f)) & = & (gP'_{\nu}(g) - P_{\nu}(g)) - (fP'_{\nu}(f) - P_{\nu}(f)). \label{eq:meanP}
\end{eqnarray}
We find
\begin{eqnarray*}
HC-MBEP & = & \Theta \left[ P_\nu(g) - P_\nu(f) - P_\mu(x) + P_\mu(f) \right] - (K_1 - x) P'_{\mu}(x) - (K_2 - f) D'(f) \\
& = & \Theta \left[ gP'_\nu(g) - x P'_\mu(x) - f D'(f) \right] - (K_1 - x) P'_{\mu}(x) - (K_2 - f) D'(f) \\
& = & \Theta \left[ (g-x) P'_\mu(x) + (g-f) D'(f) \right] - (K_1 - x) P'_{\mu}(x) - (K_2 - f) D'(f) \\
& = & P'_{\mu}(x) \left[ \Theta (g-x) - (K_1 - x) \right] + D'(f) \left[ \Theta(g-f)- (K_2-f) \right] \\
& = & 0,
\end{eqnarray*}
where we use \eqref{eq:meanP}, \eqref{eq:massP} and the definition of $\Theta$, respectively.
Optimality of the model, stopping rule and hedge now follows.

\begin{figure}[H]
\centering
\resizebox{6cm}{8cm}{
\begin{tikzpicture}[scale=0.5,
    declare function={	
    	k1=3;
	k2=1;
	a(\x)=(k1-\x)*(\x<k1);
	b(\x)=(k2-\x)*(\x<k2);
		}]

	\draw[name path=diag] (0,10) -- (10,20) node[right] {};
	
	\draw[name path=a,thick] plot[domain=0:10, samples=100] (\x,{a(\x)}) node[right] at (3,0.15) {};
	\draw[name path=b, black, thick] plot[domain=0:10, samples=100] (\x,{b(\x)}) node[right] at (3,0.15) {};
	\draw[name path=phi1, red, thick] plot[domain=0:k1, samples=100] (\x,{a(\x)+0.04}) node[right] at (3,0.15) {};
	\draw[name path=phi2, red, thick] plot[domain=k1:10, samples=100] (\x,{a(\x)+0.08}) node[right] at (3,0.15) {};
	\draw[name path=psi, black!30!green, thick] plot[domain=0:10, samples=100] (\x,{b(\x)+0.04}) node[right] at (3,0.15) {};

	\draw[name path=tu1, thick] (2, 12) to[out=85,in=225] (5.5,17.5);
	\draw[name path=tu2, thick] (5.5, 17.5) to[out=45,in=225] (10,20.5 );
	\draw[name path=td, thick] (2, 12) to[out=335,in=180] (10, 10);

\path [name path=lineA](k2,0)--(k2,20);
\draw [name intersections={of=lineA and diag}, black, dashed] (k2,0) -- (intersection-1);

\gettikzxy{(intersection-1)}{\ax}{\ay};
\draw [name path=lineB, black, dashed](\ax,\ay)--(k1,\ay);
\draw [name path=lineC, black, dashed](k1,0)--(k1,\ay);

\node (k1k2)[scale=0.5, shape=circle, fill] at (k1,10+k2) {} ;
\node [scale=1.3] (name)[below] at (k1,10+k2) {$(K_1,K_2)$} ;
\node [scale=1.3](strike1) at (k1,-0.5) {$K_1$};
\node [scale=1.3](strike2) at (k2,-0.5) {$K_2$};
\node [scale=1.3](time2)[black!30!green] at (k2-2,k2+0.5) {$\psi(y)$};
\node [scale=1.3](total)[red] at (k2-3,k1+0.5) {$\phi(x)+\psi(x)$};

\end{tikzpicture}
}
\caption{Sketch of put payoffs with $\psi(y)=(K_2-y)^+$ and $\phi(x)=(K_1-x)^+ - (K_2-x)^+$.}
\end{figure}
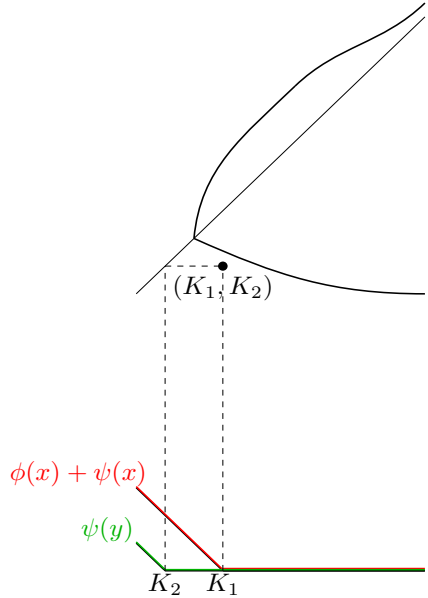

2. Now suppose $K_1 \leq e_-$. Consider an exercise rule in which the American put is exercised at time 1 if it is in-the-money, otherwise it is exercised at time 2, and a model in which mass below $K_1$ at time 1 stays constant between times 1 and 2. (This is possible since $\mu \leq \nu$ on $(-\infty, e_-)$ and $K_1 \leq e_-$.) The expected payoff of the American put is
\begin{equation}  \int_{-\infty}^{K_1} (K_1 - x) \mu(dx) + \int_{-\infty}^{K_2} (K_2 - y) (\nu - \mu)(dy) = P_\mu(K_1) + P_\nu(K_2) - P_\mu(K_2).
\label{eq:MBEP}
\end{equation}
Alternatively, suppose $K_1 > e_-$ but $f(K_1) \geq K_2$. Then under the left-curtain martingale coupling mass below $K_2$ at time 1 stays constant between times 1 and 2 (note that $K_2 \leq f(K_1) \leq e_-$), and mass between $K_2$ and $K_1$ at time 1 is mapped to $(K_2,\infty)$. Then, mass which is below $K_2$ at time 2 was either below $K_2$ at time 1, or above $K_1$ at time 1. The expected payoff under this model (using a strategy of exercising at time 1 if the American put is in-the-money) is again given by \eqref{eq:MBEP}.

Now consider the hedging cost. Let $\psi(y) = (K_2 - y)^+$. Defining $\phi$ as in Lemma~\ref{lem:phifrompsi} we find $\phi(x) = (K_1-x)^+ - (K_2-x)^+ = (K_1 -(x \vee K_2))^+$ and the superhedging cost is
\[ HC = P_\nu(K_2) + P_\mu(K_1) - P_{\mu}(K_2). \]
Hence the model-based expected payoff equals the hedging cost.
\end{proof}

\subsection[Single jump in $f$]{Two intervals of $g>x$ and one downward jump in $f$}
\label{ssec:bimodal}

We now relax the Dispersion Assumption to the case where $f$ is not monotone. The simplest situation when this may arise is when there are two intervals on which $g(x)>x$. We do not contend that there are many natural examples which fall into this situation, but rather that this intermediate case illustrates phenomena which are to be found in the general case but which were not to be found under the Dispersion Assumption.

\begin{ass}[Single Jump Assumption] \label{ass:bimodal}
$\mu$ and $\nu$ are absolutely continuous with continuous densities $\rho$ and $\eta$, respectively.
$\nu$ has support on $(\ell_\nu,r_\nu) \subseteq (-\infty,\infty)$ and $\eta>0$ on $(\ell_\nu,r_\nu)$. $\mu$ has support on $(\ell_\mu,r_\mu) \subseteq (\ell_\nu,r_\nu)$ and $\rho>0$ on $(\ell_\mu,r_\mu)$. In addition:\\
$(\mu-\nu)^+$ is concentrated on $E = (e^1_-,e^1_+)\cup(e^2_-,e^2_+)$ with $e^1_+<e^2_-$ and $\rho>\eta$ on $E$; \\
$(\nu - \mu)^+$ is concentrated on $(\ell_\nu, r_\nu) \setminus E$ and $\eta>\rho$ on $(\ell_\nu,e^1_-) \cup(e^1_+,e^2_-)\cup (e^2_+,r_\nu)$; \\
there exists $f^\prime<e^1_-$ and $x^\prime\in(e^1_+,e^2_-)$ such that
\begin{equation}
\int^{x^\prime}_{f^\prime}\mu(dz)  =  \int^{x^\prime}_{f^\prime}\nu(dz)
\hspace{10mm} \mbox{and} \hspace{10mm}
\int^{x^\prime}_{f^\prime}z\mu(dz)  =  \int^{x^\prime}_{f^\prime}z\nu(dz).
\label{eq:massmeanBimodal}
\end{equation}
\end{ass}

Under Assumption~\ref{ass:bimodal} it is possible to find $(f,g) \in \Xi^{(\ell_\mu,r_\mu),(\ell_\nu,r_\nu)}_{Mon}$. The functions $g:(\ell_\mu,r_\mu) \rightarrow (\ell_\nu,r_\nu)$ and $f:(\ell_\mu,r_\mu) \rightarrow (\ell_\nu,r_\nu)$ satisfy (see the lower part of Figure~\ref{fig:FGbimodal}):
\begin{enumerate}
\item $g(x) = x$ on $(\ell_\mu, e^1_-] \cup [x', e^2_-]$; \\
\item $g(x) > x$ on $(e^1_-,x') \cup (e^2_-,r_\mu)$; \\
\item $g$ is continuous and strictly increasing; \\
\item $f(x) = x$ on $(\ell_\mu, e^1_-] \cup [x', e^2_-]$; \\
\item $f :(e^1_-,x') \mapsto (f',x')$ is continuous and strictly decreasing; \\
\item $f :(e^2_-,r_\mu) \mapsto (\ell_\nu,e^2_-) \setminus (f',x')$ is strictly decreasing; \\
\item there exists $x'' \in (e^2_-,r_\mu)$ such that $f$ jumps at $x''$; moreover, $f(x''-)=x' > f' = f(x''+)$. Away from $x''$, $f$ is continuous on $(e_-^2,r_\mu)$.
\end{enumerate}

\begin{figure}[H]
\centering
\begin{tikzpicture}[
declare function={	
fn(\x)=(1+sin(\x r))*((-0.5*pi<\x)&&(\x<3.5*pi))+1;
phi(\x,\m,\s)=(1/sqrt(2*pi*\s))*exp(-pow(\m-\x,2)/(2*\s));
f1(\x)=0.5*phi(\x,3,0.6)+0.5*phi(\x,7,0.6);
f2(\x)=0.03+phi(\x,5,10);
diag(\x)=\x/15-0.8;
    	f=1;
	g=5;
	a=5.8;
	b=6.7;
	c=6.9;
	xmin=0; xmax=10;
	ymin=-.1; ymax=0.3;
}]
\begin{axis}[axis lines=middle,
	ytick=\empty,
            xtick=\empty,
             xmin=0, xmax=10,
             ymin=-.1, ymax=0.4,
             yticklabels={},
            xticklabels={},
            axis line style={draw=none}, clip=false]

           \addplot[name path=base,black,domain={0:10}] {0};
           \addplot[name path=f1,black,domain={0:10},samples=100] {f1(\x)};
           \addplot[name path=f2,black,domain={0:10},samples=100] {f2(\x)};

           \addplot[name path=diag,black,domain={0:10}] {diag(\x)};

\path [name path=lineA](f,0)--(f,ymax);
\draw [name intersections={of=lineA and f2}, black,] (f,0) -- (intersection-1) node[below] at (f,0)  {$f^\prime$};

\path [name path=lineA](g,0)--(g,ymax);
\draw [name intersections={of=lineA and f2}, black, ] (g,0) -- (intersection-1) node[below] at (g,0){$x^\prime=g^\prime$};

\path [name path=lineA](a,0)--(a,ymax);
\draw [name intersections={of=lineA and f2}, black] (a,0) -- (intersection-1) node[below] at (a,0) {};

\path [name path=lineA](b,0)--(b,ymax);
\draw [name intersections={of=lineA and f1}, black, ] (b,0) -- (intersection-1) node[below] at (b,0) {};

\path [name path=lineA](c,0)--(c,ymax);
\draw [name intersections={of=lineA and f2}, black, ] (c,0) -- (intersection-1) node[below] at (c,0) {};

         \addplot[pattern=crosshatch, pattern color=red!50] fill between[of=f1 and f2, soft clip={domain=f:g}];
         \addplot[pattern=crosshatch, pattern color=red!50] fill between[of=f1 and f2, soft clip={domain=a:b}];
         \addplot[pattern=crosshatch, pattern color=red!50] fill between[of=base and f2, soft clip={domain=b:c}];
           \addplot[pattern=north west lines, pattern color=blue!50] fill between[of=f1 and base, soft clip={domain=f:2}];
             \addplot[pattern=north west lines, pattern color=blue!50] fill between[of=base and f2, soft clip={domain=1.9:3.8}];
             \addplot[pattern=north west lines, pattern color=blue!50] fill between[of=base and f1, soft clip={domain=3.8:g}];
             \addplot[pattern=north west lines, pattern color=blue!50] fill between[of=base and f1, soft clip={domain=a:6.2}];
             \addplot[pattern=north west lines, pattern color=blue!50] fill between[of=base and f2, soft clip={domain=6.2:b}];

          \coordinate (A) at (2, 2/15-0.8);
            \coordinate (B) at (g, g/15-0.8);
             \draw[blue] (A) to[out=75,in=180] (B);
             \draw[blue] (A) to[out=315,in=180] (g,{diag(f)});
              \draw[blue] (6.1, {diag(6.1)}) to[out=300,in=180] (c+1,{diag(g)});
                 \draw[name path=xx, blue] (6.1, {diag(6.1)}) to[out=75,in=205] (10,{diag(10)});
                 \draw[blue] (c+1, {diag(f)}) to[out=300,in=180] (10,{diag(0)});
                 \node (f)[left, scale=0.8] at (f,{diag(f)}) {($f^\prime,f^\prime$)};
                  \node (e1)[right, scale=0.8] at (2,{diag(2)}) {(${e^1_-},{e^1_-}$)};
                  \node (g)[right, scale=0.8] at (g,{diag(g)}) {($x^\prime,x^\prime$)};
                    \node (g)[right, scale=0.8] at (6.1,{diag(6.1)}) {($e_-^2,e_-^2$)};

                  \draw[blue,dotted] (g,{diag(g)}) -- (c+1,{diag(g)});
                  \draw[blue,dotted] (c+1,{diag(g)}) -- (c+1,{diag(f)});
                \draw[blue,dotted] (g,{diag(f)}) -- (c+1,{diag(f)});
                    \draw[blue,dotted] (g,{diag(f)}) -- (c+1,{diag(f)});
                    \draw[blue,dotted] (g,{diag(g)}) -- (g,{diag(f)});

                    \path [name path=lineAA](c+1,{diag(g)})--(c+1,ymax);
\draw [name intersections={of=lineAA and xx}, gray, very thin ] (c+1,{diag(g)}) -- (intersection-1);
\node (xx)[right, scale=0.8] at (c+1,{diag(c+1)}) {($x^{\prime\prime},x^{\prime\prime}$)};
\node (gxx)[left, scale=0.8] at (intersection-1) {($x^{\prime\prime},g(x^{\prime\prime})$)};

\node (G)[above, scale=1.2] at (4,{diag(4)+0.05}) {$g$};
\node (F)[above, scale=1.2] at (4.5,{diag(4.5)-0.24}) {$f$};
\end{axis}
\end{tikzpicture}
\caption{Picture of $f$ and $g$ under Assumption~\ref{ass:bimodal}.}
\label{fig:FGbimodal}
\end{figure}

By construction we have that
\begin{equation}\label{eq:massmeanx''}
\int^{x^{\prime\prime}}_{x^\prime}\mu(dz)=\int^{g(x^{\prime\prime})}_{x^\prime}\nu(dz);
\hspace{10mm}
\int^{x^{\prime\prime}}_{x^\prime}z\mu(dz)=\int^{g(x^{\prime\prime})}_{x^\prime}z\nu(dz),
\end{equation}
so that if mass in $(x',x'')$ at time 1 is mapped to $(x',g(x''))$ at time 2 then total mass and mean are preserved. Note that given $(f',x')$ satisfy \eqref{eq:massmeanBimodal}
we also have $\int^{x^{\prime\prime}}_{f^\prime}\mu(dz)=\int^{g(x^{\prime\prime})}_{f^\prime}\nu(dz)$ and $\int^{x^{\prime\prime}}_{f^\prime}z\mu(dz)=\int^{g(x^{\prime\prime})}_{f^\prime}z\nu(dz)$.
In particular, given \eqref{eq:massmeanBimodal} and \eqref{eq:massmeanx''}, the pair of equations
\[ \int^{x^{\prime\prime}}_{f}\mu(dz)=\int^{g(x^{\prime\prime})}_{f}\nu(dz);
\hspace{10mm}
\int^{x^{\prime\prime}}_{f}z\mu(dz)=\int^{g(x^{\prime\prime})}_{f}z\nu(dz) \]
has two solutions for $f$, namely $f=x'$ and $f=f'$. Hence, in defining the left-curtain martingale coupling there are two choices for $f$ at $x''$: we may take $f(x'')=x'$ or $f(x'')=f'$.
Rather than assuming one of these choices (for example by requiring left-continuity of $f$) it is convenient to allow $f$ to be multi-valued.
Then, for $x$ such that $g(x)>x$ let $\aleph(x) = \{ f : \mbox{$(f,x,g(x))$ solves \eqref{eq:mass} and \eqref{eq:mean}} \}$.
Then, in the setting of Assumption~\ref{ass:bimodal}, for $x > e_-$, $|\aleph(x)| = 1$ except at $x''$ and there $\aleph(x'') = \{ f(x''+),f(x''-) \} = \{ f',x' \}$.

\begin{rem}
As discussed in Remark~\ref{rem:reverseconstruction} when constructing examples which fit with the analysis of this section, we may begin with $f,g$ as presented in the bottom half of Figure~\ref{fig:FGbimodal}. Given $\mu$ with support $(\ell_\mu,r_\mu)$ we can define $\nu$ via $\nu(dy) = \int \mu(dx) \chi_{f(x),x,g(x)}$. Then the pair $(\mu,\nu)$ satisfy the hypotheses of Assumption~\ref{ass:bimodal}.
\end{rem}

\begin{rem}
Recall Remark~\ref{rem:ode} and the principle that quantities in the left-curtain coupling are determined working from left to right.
Given that $\mu$ and $\nu$ have continuous densities and given that $\eta>\rho$ on $(\ell_\mu,e^1_-)$ we can set $f=g=x$ on this interval. To the right of $e^1_-$ we have $\rho>\eta$ and we can define $f$ and $g$ using the differential equations in Remark~\ref{rem:ode}. There are two cases, either $g(x)>x$ for all $x \in (e^1_-,r_\mu)$ (in which case we can define $(f,g)$ on $(e^1_-, r_\mu)$ with the properties described in Lemma~\ref{lem:FGconstruction}) or there is some point at which $g$ first hits the diagonal line $y=x$ again. This point is exactly $x'$.

If $x'$ exists it must satisfy $x' \in (e^1_+,e^2_-)$. Then we set $g(x)=x$ on $(x',e^2_-)$ and let $f=g$ solve the same coupled differential equations as in Remark~\ref{rem:ode} but with a new starting point $g(e^2_-) = e^2_- = f(e^2_-)$. The ODEs determine $f$ and $g$ until $f$ first reaches $x'$. This happens at $x''$, and at $x''$ $f$ jumps down to $f'$ (and $g$ is continuous). To the right of $x''$, $f$ and $g$ solve the differential equations again subject to initial conditions $f(x'')=f'$, $g(x'')=g(x''-)$.
\end{rem}

Recall the definition $\Lambda(x) = \frac{g(x)-K_1}{g(x)-x} - \frac{(K_1-K_2)}{x-f(x)}$. If $f$ is multi-valued, then $\Lambda$ will also be multi-valued.
In Section~\ref{ssec:dispersion}, one of our main steps was to find $x$ such that $\Lambda(x)=0$, and our aim is similar here.

Introduce $\Upsilon=\Upsilon_{K_1,K_2}(f,x,g)$ which is defined for $f\leq K_2,x  \leq K_1 \leq g$ by
\[ \Upsilon(f,x,g) = \frac{(K_2-f)-(K_1-x)}{x-f} - \frac{K_1-x}{g-x} =  \frac{g-K_1}{g-x} - \frac{(K_1-K_2)}{x-f} .  \]
Instead of seeking $x$ which is a root of $\Lambda(x)=0$ our goal is to find $(f,x,g)$ with $g=g(x)$ and $f \in \aleph(x)$ such that $\Upsilon(f,x,g)=0$.

Fixing $K_1$, the value of $K_2$ such that $\Upsilon(f'=f(x''+),x'',g(x''))=0$ is given by
$K_2 =f'+(K_1-x^{\prime\prime})\frac{g(x^{\prime\prime})-f'}{g(x^{\prime\prime})-x^{\prime\prime}}$. Similarly,
the value of $K_2$ such that $\Upsilon(x'=f(x''-),x'',g(x''))=0$ is given by
$K_2 =x'+(K_1-x^{\prime\prime})\frac{g(x^{\prime\prime})-x'}{g(x^{\prime\prime})-x^{\prime\prime}}$.
This motivates the introduction of the linear increasing functions $L_u$, $L_d:[x^{\prime\prime},g(x^{\prime\prime})]\to\mathbb{R}$ defined by
\begin{align}
L_u(x)&=x^\prime+(x-x^{\prime\prime})\frac{g(x^{\prime\prime})-x^\prime}{g(x^{\prime\prime})-x^{\prime\prime}},\\
L_d(x)&=f'+(x-x^{\prime\prime})\frac{g(x^{\prime\prime})-f'}{g(x^{\prime\prime})-x^{\prime\prime}}.
\end{align}
Pictorially, $L_d$ and $L_u$ are the lower and upper boundaries, respectively, of the dotted triangular area $\sG$ in Figure~\ref{fig:regionsBimodal}.

From Figure \ref{fig:regionsBimodal} we identify four regions (and various subregions) on which four different martingale couplings and hedging strategies will be needed in order to find the highest model-based expected payoff of an American put. (Compare this with two regimes under the Dispersion Assumption in Figure \ref{fig:fgdispersion}.) 

Define 
\[ \sR_1  =  \{ (k_1,k_2): e^1_- < k_1 < x', f(k_1) < k_2 < k_1 \}, \]
which we write more compactly as
$\sR_1  = \{ e^1_- < k_1 < x', f(k_1) < k_2 < k_1 \}$.
Using the same compact notation define
\begin{eqnarray*}
\sR_1 & = & \{ e^1_- < k_1 < x', f(k_1) < k_2 < k_1 \}; \\
\sR_2 & = & \{ e^2_- < k_1 < x'', f(k_1) < k_2 < k_1 \} \cup \{ k_1 = x'' , x' < k_2 < k_1 \}; \\
\sR_3 & = & \{ x'' < k_1 < g(x''), L_u(k_1) \leq k_2 < k_1 \}; \\
\sR_4 & = & \{ x'' < k_1 < g(x''), f(k_1) < k_2 \leq L_d(k_1) \}; \\
\sR_5 & = & \{ g(x'') \leq k_1 \leq r_\mu, f(k_1) < k_2 < k_1 \}; \\
\sB_1 & = & \{ \ell_\mu \leq k_1 \leq e^1_-, k_2 < k_1 \} \cup \{ e^1_- < k_1 < x', k_2 \leq f(k_1) \} \cup \{ x'' < k_1 \leq r_\mu, k_2 \leq f(k_1) \};  \\
\sB_2 & = & \{ x' \leq k_1 \leq x'', k_2 \leq f' \}; \\
\sB_3 & = & \{ x' \leq k_1 \leq e^2_-, x' \leq k_2 < k_1 \} \cup \{ e^2_- < k_1 \leq x'', x' \leq k_2 \leq f(k_1) \}; \\
\sG & = & \{ x'' < k_1 < g(x''), L_d(k_1) < k_2 < L_u(k_1) \}; \\
\sW & = & \{ x' \leq k_1 \leq x'', f' < k_2 < x' \};
\end{eqnarray*}
and set $\sR= \cup_{i=1}^5 \sR_i$ and $\sB= \cup_{i=1}^3 \sB_i$. In general, on the boundaries between the regions the boundaries could be allocated to either region. However, we allocate points on the boundary to the region where the hedge is simplest.

Note that $\sR \cup \sB \cup \sG \cup \sW = \{ (k_1,k_2) : \ell_\mu \leq k_1 \leq r_\mu, k_2 < k_1 \}$.

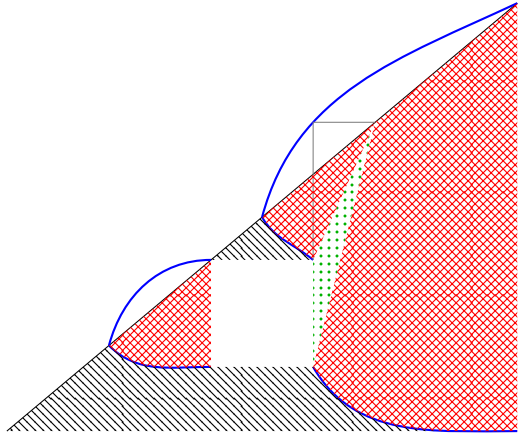
\begin{figure}[H]
\centering

\begin{tikzpicture}[
declare function={	
diag(\x)=\x;
    	f=1.5;
	e1=2;
	g=4;
	e2=5;
	f2=6;
	s=8.5;
}]
\begin{axis}[axis lines=middle,
	ytick=\empty,
            xtick=\empty,
             xmin=-.1, xmax=10,
             ymin=0, ymax=10,
             yticklabels={},
            xticklabels={},
            axis line style={draw=none}, clip=false]

           \addplot[name path=diag,black,domain={0:10}] {diag(\x)};

                \path[name path=base] (0,0) to (10,0);


          \coordinate (g) at (g, {diag(g)});
           \coordinate (e1) at (e1, {diag(e1)});
            \coordinate (e2) at (e2, {diag(e2)});
             \coordinate (s) at (s, {diag(s)});

             \coordinate (a) at (g, {diag(f)});
              \coordinate (b) at (f2, {diag(g)});
               \coordinate (c) at (f2, {diag(f)});

             \draw[thick, blue, name path=tu1] (e1) to[out=75,in=180] (g);
             \draw[thick, blue, name path=td1] (e1) to[out=315,in=180] (a);
              \draw[thick, blue, name path=tu2] (e2) to[out=75,in=205] (10,{diag(10)});
              \draw[thick, blue, name path=td2] (e2) to[out=300,in=140] (b);
              \draw[thick, blue, name path=td3,] (c) to[out=300,in=180] (10,{diag(0)});

                \path[gray, name path=dottop] (g) to (b);
                 \path[gray, name path=dotright] (b) to (c);
                  \path[gray, name path=dotbottom] (a) to (c);
                   \path[gray, name path=dotleft] (g) to (a);

		\path [gray, dotted, pattern=dots, pattern color=white] (g) -- (b) -- (c) -- (a) -- (g);

                  \path [name path=lineA](b) -- (f2,10);
		\draw [gray, name intersections={of=lineA and tu2}] (b) -- (intersection-1)%
		\pgfextra{\pgfgetlastxy{\macrox}{\macroy} \global\let\macroy\macroy};
         \coordinate (z) at (intersection-1);
         \path [name path=lineA](intersection-1) -- (10,\macroy);
         \draw [gray,name intersections={of=lineA and diag}] (z) -- (intersection-1);

		\path [thick, green, pattern=dots, pattern color=black!30!green] (f2,g) -- (intersection-1) -- (f2,f) -- (f2,g);

		
		 \path [pattern=north west lines, pattern color=black] (0,0) -- (e1) to[out=315,in=180] (a) -- (c) to[out=300,in=180] (10,{diag(0)}) -- (0,0);
	
		\path [ pattern=north west lines, pattern color=black] (g) -- (e2) to[out=300,in=140] (b) -- (g);

		
                            \path [pattern=crosshatch,pattern color=red] (b) -- (intersection-1) -- (e2) to[out=300,in=140] (b);
                            \path [pattern=crosshatch,pattern color=red] (10,0) -- (10,10) -- (intersection-1) -- (c) to[out=300,in=180] (10,{diag(0)});

                	 \path [pattern=crosshatch,pattern color=red] (e1) to[out=315,in=180] (a) -- (g) -- (e1);

             \end{axis}
\end{tikzpicture}
\caption{Picture of $f$ and $g$ in bi-modal case, now with 4 regions shaded (cross-hatched, diagonally, dotted and blank).}
\label{fig:regionsBimodal}
\end{figure}

\subsubsection{Case $(K_1,K_2) \in \sR$.}
\label{sssec:R}

\begin{lem}
Suppose $(K_1,K_2) \in \sR$. Then there exists $x^*=x^*(\mu,\nu;K_1,K_2) \in (g^{-1}(K_1),K_1)$ and $f^* \in \aleph(x^*)$ such that $\Upsilon(f^*,x^*,g^*=g(x^*))=0$.
\label{lem:x^*Red}
\end{lem}

\begin{proof}
Suppose $(K_1,K_2) \in \sR_1 \cup \sR_2 \cup \sR_5$. Consider $\Lambda:[g^{-1}(K_1),K_1] \mapsto \R$ defined by (\ref{eq:Lambdadef}). Note that for this choice of $(K_1,K_2)$, $f$ and $g$ are both continuous on $[g^{-1}(K_1),K_1]$, see Figure~\ref{fig:FGbimodal}. Hence $\Lambda(x)=\Upsilon(f(x),x,g(x))$ is also continuous. Then the same argument as in the proof of Lemma \ref{lem:uniquex*} shows that there exists a unique $x^*=x^*(\mu,\nu;K_1,K_2) \in (g^{-1}(K_1),K_1)$ such that $\Lambda(x^*)=0$.

Now suppose $(K_1,K_2) \in \sR_3 \cup \sR_4$ and consider $\Lambda$ as before. Recall that $\Lambda$ is increasing, $\Lambda(g^{-1}(K_1))<0$ and $\Lambda(K_1)>0$. On the other hand, $g^{-1}(K_1)<x^{\prime\prime}$ and hence $\Lambda$ has an upward jump at $x^{\prime\prime}$ (since $f$ has a downward jump at $x^{\prime\prime}$). There are two cases depending on whether $(K_1,K_2) \in \sR_3$ or $\sR_4$.
\begin{enumerate}
\item Suppose $K_2>L_u(K_1)$. Then $\Lambda(x^{\prime\prime}-)>0$. Since $\Lambda(g^{-1}(K_1))<0$, by the continuity of $\Lambda$ on $(g^{-1}(K_1),x^{\prime\prime})$, there exists a unique $x^*=x^*(\mu,\nu;K_1,K_2) \in (g^{-1}(K_1),x^{\prime\prime})$ such that $\Lambda(x^*)=0$.
    If $K_2 = L_u(K_1)$ then $\Upsilon(x',x''g(x''))=0$ and we take $x^* = x''$, $g^* = g(x'')$ and $f^* = f(x''-) = x'$.
\item Suppose $K_2<L_d(K_1)$. Then $\Lambda(x^{\prime\prime}+)<0$. Further, since $\Lambda(K_1)>0$ there exists a unique $x^*=x^*(\mu,\nu;K_1,K_2) \in (x^{\prime\prime},K_1)$ such that $\Lambda(x^*)=0$.
    If $K_2 = L_d(K_1)$ then $\Upsilon(f',x'',g(x''))=0$ and we take $x^* = x''$, $g^* = g(x'')$ and $f^* = f(x''+) = f'$.
\end{enumerate}
\end{proof}

By Lemma~\ref{lem:x^*Red}, for $(K_1,K_2) \in \sR$ there exists $\{ f^* \in \aleph(x^*),x^*, g^*=g(x^*) \}$ such that $\Upsilon(f^*,x^*,g^*)=0$. Suppose $(K_1,K_2) \in \sR_1 \cup \sR_4 \cup \sR_5$. As before, let $M^*=(M^*_0=\bar{\mu},M^*_1=X,M^*_2=Y)$ be the stochastic process such that $\Prob(X \in dx, Y \in dy) = \hat{\pi}_{x^*}(dx,dy)$, where $\hat{\pi}_{x^*} \in \hat{\Pi}(\mu,\nu)$ is a martingale coupling that combines the couplings $\mu_{f^*,x^*} \mapsto \nu_{f^*,g^*}$ and $\tilde{\mu}_{f^*,x^*} \mapsto\tilde{\nu}_{f^*,g^*}$.

Recall the proof of Theorem \ref{thm:dispersion}. There, to show that $MBEP=HC$, we used the fact that  under model $M^*$, or more specifically, under any martingale coupling which mapped $(f^*,x^*)$ to $(f^*,g^*)$, the mass that is `unexercised' at time 1 and is in-the-money at time 2 is given by $(\nu-\mu)\lvert_{(-\infty,f^*)}$ where $f^*<e_-$. When $f(x')<e_-^1$ (as is the case when $(K_1,K_2) \in \sR_1 \cup \sR_4 \cup \sR_5$) then the same proof applies, $MBEP=HC$ and we have optimality.
However, if $(K_1,K_2) \in \sR_2 \cup \sR_3$, then it is not the case that $f^*<e^1_-$  and thus, in order to specify the optimal model, we need to impose additional structure on the coupling $\tilde{\mu}_{f^*,x^*} \mapsto\tilde{\nu}_{f^*,g^*}$.

Suppose $(K_1,K_2) \in \sR_2 \cup \sR_3$. Then $x^\prime<f^*$, so that $(f^\prime,x^\prime)\cap(f^*,g^*)=\emptyset$. From the defining properties of $f^\prime$ and $x^\prime$ we see that there exists a martingale coupling, which we term $\hat{\pi}_{x^\prime,x^*} \in \hat{\Pi}_M(\mu,\nu)$, that combines the couplings of $\mu_{{f^\prime,x^\prime}}\mapsto\nu_{{f^\prime,x^\prime}}$ and $\mu_{f^*,x^*} \mapsto \nu_{f^*,g^*}$, so that $\hat{\pi}_{x^\prime,x^*}$ maps $(f^\prime,x^\prime)$ to $(f^\prime,x^\prime)$, $(f^*,x^*)$ to $(f^*,g^*)$ and $((f^\prime,x^\prime)\cup (f^*,x^*))^C$ to $((f^\prime,x^\prime)\cup (f^*,g^*))^C$. Let $\hat{M}^*=(\hat{M}^*_0=\bar{\mu},\hat{M}^*_1=X,\hat{M}^*_2=Y)$ be the stochastic process such that $\Prob(X \in dx, Y \in dy) = \hat{\pi}_{x^\prime,x^*}(dx,dy)$.

If $(K_1,K_2) \in \sR_1 \cup \sR_4 \cup \sR_5$ we have $M^* \in \sM(\sS^*, \mu, \nu)$ and if $(K_1,K_2) \in \sR_2 \cap \sR_3$ we have that $\hat{M}^*\in \sM(\sS^*, \mu, \nu)$. For both models we consider a candidate stopping time $\tau^*=1$ if $X<x^*$ and $\tau^*=2$ otherwise, and a candidate superhedge $(\psi^*, \phi^*,\theta_{1,2}^*)$ generated by the function $\psi^*$ defined in (\ref{eq:psi*}).
\begin{thm}
Suppose Assumption~\ref{ass:bimodal} holds and $(K_1,K_2) \in \sR$.
Depending on whether $(K_1,K_2) \in \sR_1 \cup \sR_4 \cup \sR_5$ or $\sR_2 \cup \sR_3$, the models $M^*$ and $\hat{M}^*$ and the stopping time $\tau^*$ are the consistent models for which the price of the American option is the highest. The function $\psi^*$ defined in \eqref{eq:psi*} defines the cheapest superhedge. Moreover, the highest model-based price is equal to the cost of the cheapest superhedge.
\label{thm:bimodalRed}
\end{thm}
\begin{proof}
If $(K_1,K_2) \in \sR_1 \cup \sR_4 \cup \sR_5$ then the proof is essentially the same as the proof of the first case in Theorem~\ref{thm:dispersion}. We repeat it for convenience. First find $x^*\in(g^{-1}(K_1),K_1)$ and $f^* \in \aleph(x^*)$ such that $\Upsilon(f^*,x^*g^*=g(x^*))=0$. If $x^*=x''$ we find $f^* = f(x''+)=f'$. Under the candidate model $M^*$ mass below $f^*$ at time 1 is mapped to the same point at time 2 (which is possible since $f^*<e^1_-$), and mass in $(f^*,x^*)$ is mapped to $(f^*,g^*)$, while mass above $x^*$ is either mapped to below $f^*$ or to above $g^*$. Then under the candidate stopping rule $\tau^*$ the model-based expected payoff is equal to the cost of the hedging strategy generated by $\psi^*$:
\begin{align*}
MBEP & =\lefteqn{ \int_{-\infty}^{x^*} (K_1-w)^+ \mu(dw) + \int_{-\infty}^{f^*} (K_2 - w)^+ (\nu - \mu)(dw) } \\
& =   P_\mu(x^*) + (K_1 - x^*) P'_{\mu}(x^*) + D(f^*) + (K_2 - f^*) D'(f^*)\\
& = HC.
\end{align*}

Now suppose $(K_1,K_2) \in \sR_2 \cup \sR_3$. By Lemma~\ref{lem:x^*Red} there is a unique $x^*\in (g^{-1}(K_1),x^{\prime\prime}]$ and $f^* \in \aleph(x^*)$ such that $\Upsilon(f^*,x^*,g^*=g(x^*))=0$.
If $x^*=x''$ then we have $f^*=f(x''-) = x'$.
Then, since $\nu$ is continuous we have that $f^*,x^*,g^*$ solve \eqref{eq:mass} and \eqref{eq:mean}. Note, however, that $x^\prime \leq f^*<e^2_-$.

Under the candidate model $\hat{M}^*$ mass in $(f^\prime,x^\prime)$ at time 1 is mapped to the same interval at time 2. Also, mass below $f^\prime$ and mass in $(x^\prime,f^*)$ at time 1 is mapped to the same point at time 2, and mass in $(f^*,x^*)$ is mapped to $(f^*,g^*)$. Mass above $x^*$ is either mapped to below $f^\prime$, to $(x^\prime,f^*)$, or to above $g^*$. In particular $(\nu - \mu)\lvert_{(-\infty,f^\prime)\cup(x^\prime,f^*)}$ is the mass that was not `exercised' at time $1$ and is `exercised' in-the-money at time 2. In other words, $(\nu - \mu)\lvert_{(-\infty,f^\prime)\cup(x^\prime,f^*)}$ is the probability under $\hat{M}^*$ that $(X>x^*,Y<K_2)$. From \eqref{eq:massmeanBimodal} we have $\int_{x'}^{f'} (K_2-w) (\nu - \mu)(dw)=0$. Then
\begin{align*}
MBEP & = \int_{-\infty}^{x^*} (K_1-w) \mu(dw) + \int_{-\infty}^{f^\prime} (K_2 - w) (\nu - \mu)(dw) +\int_{x^\prime}^{f^*} (K_2 - w) (\nu - \mu)(dw)\\
&  = \int_{-\infty}^{x^*} (K_1-w) \mu(dw) + \int_{-\infty}^{f^*} (K_2 - w) (\nu - \mu)(dw) - \int_{f^\prime}^{x^\prime} (K_2 - w) (\nu - \mu)(dw) \\
& = \int_{-\infty}^{x^*} (K_1-w) \mu(dw) + \int_{-\infty}^{f^*} (K_2 - w) (\nu - \mu)(dw) \\
& =   P_\mu(x^*) + (K_1 - x^*) P'_{\mu}(x^*) + D(f^*) + (K_2 - f^*) D'(f^*) \\
& = HC.
\end{align*}
\end{proof}

\subsubsection{$(K_1,K_2) \in \sB = \sB_1 \cup \sB_2 \cup \sB_3$}
\label{sssec:B}

\begin{thm}
Suppose Assumption~\ref{ass:bimodal} holds and that $(K_1,K_2) \in \sB = \sB_1 \cup \sB_2 \cup \sB_3$.
Then there is a consistent model for which $(Y<K_2) = (X<K_2)\cup (X>K_1, Y< K_2)$ and, if $x^\prime<K_2$, $(f^\prime<X<x^\prime)=(f^\prime<Y<x^\prime)$. Then any model with these properties with the stopping rule $\tau=1$ if $X<K_1$ and $\tau=2$ otherwise attains the highest consistent model price. The cheapest superhedge is generated from ${\psi}(x) = (K_2-x)^+$ and the highest model-based price is equal to the cost of the cheapest hedge.
\label{thm:bimodalBlack}
\end{thm}
\begin{proof}
Let $\psi(y) = (K_2 - y)^+$. Defining $\phi$ as in Lemma~\ref{lem:phifrompsi} we find $\phi(x) = (K_1-x)^+ - (K_2-x)^+$ and the superhedging cost (which is the same for all the cases) is
\[ HC = P_\nu(K_2) + P_\mu(K_1) - P_{\mu}(K_2). \]

Suppose $(K_1,K_2) \in \sB_1$. Then using the properties of $f$ and $g$ and the left-curtain coupling we see that the proof that the model-based expected payoff is equal to the hedging cost is the same as in the second case of Theorem \ref{thm:dispersion}. In particular,
\begin{align*}
MBEP=\int_{-\infty}^{K_1} (K_1 - x) \mu(dx) + \int_{-\infty}^{K_2} (K_2 - y) (\nu - \mu)(dy)= P_\mu(K_1) + P_\nu(K_2) - P_\mu(K_2).
\end{align*}

Now suppose $(K_1,K_2) \in \sB_2$. Then under the left-curtain coupling mass from $(f^\prime,x^\prime)$ at time $1$ is mapped to the same interval at time $2$. Therefore mass which is below $K_2$ at time $2$ was either below $K_2$ at time $1$, or above $x^\prime$ at time $1$. Therefore, we again have
\begin{align*}
MBEP=\int_{-\infty}^{K_1} (K_1 - x) \mu(dx) + \int_{-\infty}^{K_2} (K_2 - y) (\nu - \mu)(dy)= P_\mu(K_1) + P_\nu(K_2) - P_\mu(K_2).
\end{align*}

Finally, suppose $(K_1,K_2) \in \sB_3$.
We again utilise the fact that under the left-curtain coupling, mass from $(f^\prime,x^\prime)$ at time $1$ is mapped to the same interval at time $2$. In both cases, the mass which is below $K_2$ at time $2$ was either below $K_2$ at time $1$, or above $K_1$ at time $1$. In particular, mass that can be `exercised' at time $2$ is given by $(\nu-\mu)\lvert_{(-\infty,f^\prime)\cup(x^\prime,K_2)}$. Then using $\int^{x^\prime}_{f^\prime}(K_2-z)(\nu-\mu)(dz)=0$ we again have
\begin{align*}
MBEP &=\int_{-\infty}^{K_1} (K_1 - x) \mu(dx) + \int_{-\infty}^{f^\prime} (K_2 - y) (\nu - \mu)(dy) + \int_{x^\prime}^{K_2} (K_2 - y) (\nu - \mu)(dy)\\
&=\int_{-\infty}^{K_1} (K_1 - x) \mu(dx) + \int_{-\infty}^{K_2} (K_2 - y) (\nu - \mu)(dy),
\end{align*}
which ends the proof.
\end{proof}

\subsubsection{$(K_1,K_2) \in \sW$}
\label{sssec:W}

Suppose $(K_1,K_2) \in \sW$. For this case we associate the following superhedge:
let $\psi^{x^\prime}$ be given by
\begin{equation}
\psi^{x^\prime}(z) = \left\{ \begin{array}{ll} (K_2 - z)    &   z \leq f^\prime; \\
                                   (K_2-f^\prime)-(z-f^\prime)\frac{K_2-f^\prime}{x^\prime-f^\prime} & f^\prime < z \leq x^\prime; \\
                                    0 & z > x^\prime,
                                    \end{array} \right.
\label{eq:psiPrime0}
\end{equation}
see Figure \ref{fig:dotted}. Since $\psi^{x^\prime}$ is convex and $\psi^{x^\prime}(z)\geq (K_2-z)^+$, we can use Lemma \ref{lem:phifrompsi} to generate a corresponding superhedging strategy $(\psi^{x^\prime}, \phi^{x^\prime}, \theta_{1,2}^{x^\prime})$.

\begin{thm}
Suppose Assumption~\ref{ass:bimodal} holds and $(K_1,K_2)\in\sW$.\\
Then there is a consistent model for which $(f' <X<x') = (f'<Y<x')$ and any model with this property with the stopping rule $\tau=1$ if $X<K_1$ and $\tau=2$ otherwise attains the highest consistent model price. The cheapest superhedge is generated from ${\psi}^{x^\prime}$ defined in (\ref{eq:psiPrime0}) and the highest model-based price is equal to the cost of the cheapest hedge.
\label{thm:dotted}
\end{thm}

\begin{proof}
 First note that
\begin{equation*}
\psi^{x^\prime}(z)=\Theta(x^\prime-z)^+ +(1-\Theta)(f^\prime-z)^+,
\end{equation*}
where $\Theta = \frac{K_2 - f'}{x'-f'}$.
Since $x^\prime<K_1$ we have
\begin{eqnarray}
\phi^{x^\prime}(w)+\psi^{x^\prime}(z) & = &(K_1-w)^+-\psi^{x^\prime}(w)+\psi^{x^\prime}(z)\nonumber\\
& = & (K_1-w)^+ + \Theta[(x^\prime-z)^+ - (x^\prime-w)^+] +(1-\Theta)[(f^\prime-z)^+-(f^\prime-w)^+] \nonumber
\end{eqnarray}
and the cost of this strategy (under any consistent model) is $HC=P_\mu(K_1)  +\Theta D(x^\prime)+(1-\Theta)D(f^\prime)$.

\begin{figure}[H]

\centering
\resizebox{7 cm}{10cm}{
\begin{tikzpicture}[scale=1,
declare function={	
diag(\x)=\x;
    	f=1.5;
	e1=2;
	g=4;
	e2=5;
	f2=6;
	s=8.5;
	k1=5;
	k2=3;
	a(\x)=(k1-\x)*(\x<k1)-6;
	b(\x)=(k2-\x)*(\x<k2)-6;
}]

           \draw[name path=diag,black] (0,0) -- (10,10);


          \coordinate (g) at (g, {diag(g)});
           \coordinate (e1) at (e1, {diag(e1)});
            \coordinate (e2) at (e2, {diag(e2)});
             \coordinate (s) at (s, {diag(s)});

             \coordinate (a) at (g, {diag(f)});
              \coordinate (b) at (f2, {diag(g)});
               \coordinate (c) at (f2, {diag(f)});

             \draw[very thick, blue, name path=tu1] (e1) to[out=75,in=180] (g);
             \draw[very thick, blue, name path=td1] (e1) to[out=315,in=180] (a);
              \draw[very thick, blue, name path=tu2] (e2) to[out=75,in=205] (10,{diag(10)});
              \draw[very thick, blue, name path=td2] (e2) to[out=300,in=140] (b);
              \draw[very thick, blue, name path=td3,] (c) to[out=300,in=180] (10,{diag(0)});

		 \draw [gray, densely dotted, very thick] (g) -- (b) -- (c) -- (a) -- (g);

                   \path[name path=base] (0,0) to (10,0);

                  \path [name path=lineA](b) -- (f2,10);
		\path [gray, very thin, name intersections={of=lineA and tu2}] (b) -- (intersection-1)%
		\pgfextra{\pgfgetlastxy{\macrox}{\macroy} \global\let\macroy\macroy};
         \coordinate (z) at (intersection-1);
         \path [name path=lineB](intersection-1) -- (10,\macroy);
         \path [gray, very thin, name intersections={of=lineB and diag}] (z) -- (intersection-1);
		
		\draw [black] (f2,g) -- (intersection-1) -- (f2,f) ;
		
		\node (k1k2)[scale=0.7, shape=circle, fill] at (k1,k2) {} ;
\node (name)[below,scale=1.5] at (k1,k2) {$(K_1,K_2)$} ;
\node () [below,scale=1.5] at (f,{b(g)}) {$f^\prime$} ;
\node () [below,scale=1.5] at (g,{b(g)}) {$x^\prime$} ;


	\draw[name path=a, black, thick] plot[domain=0:10, samples=100] (\x,{a(\x)}) ;
	\draw[name path=b, black, thick] plot[domain=0:10, samples=100] (\x,{b(\x)}) ;

                          \draw [gray, very thin, loosely dashed] (f,f) -- (a) ;
                        \draw [gray, very thin, loosely dashed] (f,f) -- (f,{b(g)}) ;
                       \draw [gray, very thin] (k2,k2) -- (k2,{b(k2)}) ;
                            \draw [gray, very thin] (k2,k2) -- (k1,k2) ;
                               \draw [gray, very thin] (k1,k2) -- (k1,{a(k1)}) ;
                                \draw [gray, very thin, loosely dashed] (g,{b(g)})--(a);

	\draw[black!30!green, very thick] plot[domain=0:f, samples=100] (\x,{b(\x)+0.04});
	\draw[black!30!green, very thick] plot[domain=g:10, samples=100] (\x,{b(\x)+0.04});
	\draw [black!30!green, very thick] (f,{b(f)+0.04}) -- (g,{b(g)+0.04}) ;
	
	\draw[red] plot[domain=0:k1, samples=100] (\x,{a(\x)+0.04});
	\node (hedge)[left,black!30!green,scale=2] at (0,{b(0)}) {$\psi^{x^\prime}$} ;

\end{tikzpicture}
}
\caption{Picture of $f$ and $g$ along with superhedge for the blank region $\sW$. }
\label{fig:dotted}
\end{figure}
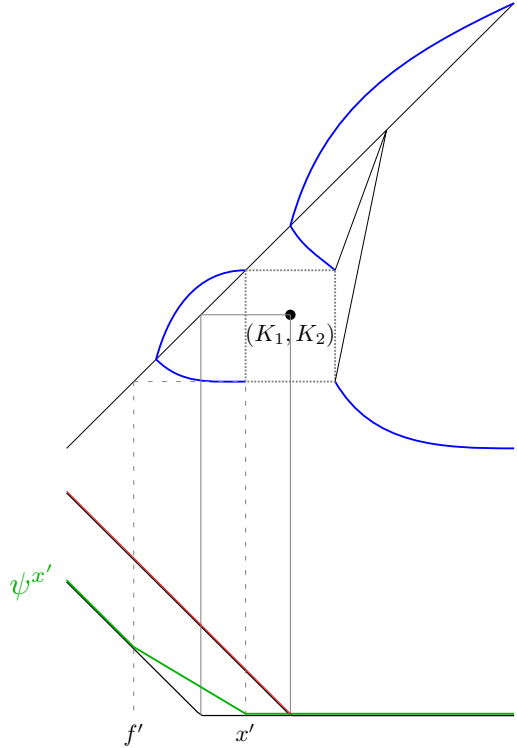

Now consider the model-based expected payoff. From (\ref{eq:massmeanBimodal}) it follows that $\mu_{{f^\prime,x^\prime}}$ and $\nu_{{f^\prime,x^\prime}}$ have the same mean and mass, and are in convex order. Moreover, the same holds for $\tilde{\mu}_{{f^\prime,x^\prime}}$ and $\tilde{\nu}_{{f^\prime,x^\prime}}$. Hence there exists a martingale coupling, which we term $\hat{\pi}^{x^\prime}\in\hat{\Pi}_M(\mu,\nu)$, which maps the mass in $(f^\prime,x^\prime)$ at time $1$ to the same interval at time $2$. Under this model the only mass that can be `exercised' at time $2$ is therefore given by $(\nu-\mu)\lvert_{(-\infty,f^\prime)}$.

Note that, since $f^\prime$ and $x^\prime$ satisfy (\ref{eq:massmeanBimodal}), and hence $\int_{f'}^{x'} (x'-w) (\nu - \mu)(dw)=0$,
\begin{eqnarray}
D(x^\prime)-D(f^\prime) & = & \int_{-\infty}^{x^\prime} (x^\prime-w)^+ (\nu-\mu)(dw) - \int_{-\infty}^{f^\prime} (f^\prime - w)^+ (\nu - \mu)(dw)\nonumber\\
& = & \int_{-\infty}^{f^\prime} (x^\prime-f^\prime) (\nu-\mu)(dw) + \int_{f^\prime}^{x^\prime} (x^\prime - w) (\nu-\mu)(dw)\nonumber\\
& = & (x^\prime-f^\prime) \int_{-\infty}^{f^\prime}  (\nu-\mu)(dw). \nonumber
\end{eqnarray}
Then given that we stop at time $1$ if $X<K_1$ and at time $2$ otherwise we have
\begin{align*}
MBEP&= \int_{-\infty}^{K_1} (K_1-w)^+ \mu(dw) + \int_{-\infty}^{f^\prime} (K_2 - w)^+ (\nu - \mu)(dw)\\
& = \int_{-\infty}^{K_1} (K_1-w) \mu(dw) + \int_{-\infty}^{f^\prime} (f^\prime - w) (\nu - \mu)(dw) + (K_2- f^\prime) \int_{-\infty}^{f^\prime} (\nu-\mu)(dw)  \\
& =  P_\mu(K_1) +D(f^\prime)+ \Theta[D(x^\prime)-D(f^\prime)] \\
& = P_\mu(K_1) + \Theta D(x^\prime)+(1-\Theta)D(f^\prime) = HC
\end{align*}
as required.
\end{proof}

\subsubsection{$(K_1,K_2) \in \sG$}
\label{sssec:G}

Recall the construction of $L_u$ and $L_d$. For $K_1\in(x^{\prime\prime},g(x^{\prime\prime}))$ and $K_2\in(L_d(K_1),L_u(K_1))$ there does not exist $x^*\in(g^{-1}(K_1),K_1)$ such that $\Lambda(x^*)=0$; instead $\Lambda(x''-) < 0 < \Lambda(x''+)$.  On the other hand, from (\ref{eq:massmeanx''}) we have that there exists a martingale coupling of $\mu_{{x^\prime,x^{\prime\prime}}}$ and $\nu_{{x^{\prime},g(x^{\prime\prime})}}$. Moreover, note that the restrictions of $\tilde{\mu}_{{f^\prime,x^\prime}}$ to $(x^\prime,x^{\prime\prime})$ and $\tilde{\nu}_{{f^\prime,x^\prime}}$ to $(x^\prime,g(x^{\prime\prime}))$ are equal to $\mu_{{x^\prime,x^{\prime\prime}}}$ and $\nu_{{x^{\prime},g(x^{\prime\prime})}}$, respectively. Then we define a martingale coupling $\hat{\pi}^{x^\prime,x^{\prime\prime}} \in \hat{\Pi}(\mu,\nu)$ by combining the couplings of $\mu_{{f^\prime,x^\prime}} \mapsto \nu_{{f^\prime,x^\prime}}$, $\mu_{{x^\prime,x^{\prime\prime}}} \mapsto \nu_{{x^\prime,g(x^{\prime\prime})}}$ and $(\tilde{\mu}_{{f^\prime,x^\prime}}-\mu_{{x^\prime,x^{\prime\prime}}}) \mapsto(\tilde{\nu}_{{f^\prime,x^\prime}}-\nu_{{x^\prime,g(x^{\prime\prime})}})$. In other words, $\hat{\pi}^{x^\prime,x^{\prime\prime}}$ maps $({f^\prime,x^{\prime}})$ to $({f^\prime,x^{\prime}})$, $({x^\prime,x^{\prime\prime}})$ to $({x^\prime,g(x^{\prime\prime})})$ and $(({f^\prime,x^{\prime}})\cup ({x^\prime,x^{\prime\prime}}))^C$ to $(({f^\prime,x^{\prime}})\cup ({x^\prime,g(x^{\prime\prime})}))^C$. Let $M^{x^\prime,x^{\prime\prime}}=(M^{x^\prime,x^{\prime\prime}}_0=\bar{\mu},M^{x^\prime,x^{\prime\prime}}_1=X,M^{x^\prime,x^{\prime\prime}}_2=Y)$ be the stochastic process such that $\Prob(X \in dx, Y \in dy) = \hat{\pi}^{x^\prime,x^{\prime\prime}}(dx,dy)$.
Then $M^{x^\prime,x^{\prime\prime}} \in \sM(\sS^*, \mu, \nu)$.  Note that the model $M^{x^\prime,x^{\prime\prime}}$ is a refinement of $M^{x^\prime}$ used in the proof of Theorem \ref{thm:dotted}.

Given $x^\prime$, and thus also $x^{\prime\prime}$, we define the superhedge as follows. First define linear functions $\Delta_1:[f^\prime,x^\prime]\to\mathbb{R}$ and $\Delta_2:[x^\prime,g(x^{\prime\prime})]\to\mathbb{R}$ by
\begin{equation*}
\Delta_1(x)=(K_2-f^\prime)-(x-f^\prime)\frac{(K_2-f^\prime)-\Delta_2(x^\prime)}{x^\prime-f^\prime};
\hspace{10mm}
\Delta_2(x)=(g(x^{\prime\prime})-x)\frac{K_1-x^{\prime\prime}}{g(x^{\prime\prime})-x^{\prime\prime}}
\end{equation*}
and note that $\Delta_1(f^\prime)=(K_2-f^\prime)$, $\Delta_1(x^\prime)=\Delta_2(x^\prime)$, $\Delta_2(x'')=K_1-x''$ and $\Delta_2(g(x^{\prime\prime}))=0$. Moreover, direct calculation shows that $-1<\Delta^\prime_1(x)<\Delta^\prime_2(x)<0$. Now define a function $\psi^{x^\prime,x^{\prime\prime}}$ by
\begin{equation}
\psi^{x^{\prime},x^{\prime\prime}}(z) = \left\{ \begin{array}{ll} (K_2 - z)    &   z \leq f^\prime; \\
                                   \Delta_1(z) & f^\prime < z \leq x^\prime; \\
                                    \Delta_2(z) & x^\prime < z \leq g(x^{\prime\prime});\\
                                    0 & z > g(x^{\prime\prime}).
                                    \end{array} \right.
\label{eq:psiPrime}
\end{equation}

\begin{figure}[H]
\centering
\resizebox{8 cm}{11cm}{
\begin{tikzpicture}[ dot/.style={circle,inner sep=1pt, fill, label={#1},name#1},
extended line/.style={shorten >=-#1, shorten <=-#1},
extended line/.default=3cm,
declare function={	
diag(\x)=\x;
    	f=1.5;
	e1=2;
	g=4;
	e2=5;
	f2=6;
	s=8.5;
	k1=6.37;
	k2=4.5;
	a(\x)=(k1-\x)*(\x<k1)-8;
	b(\x)=(k2-\x)*(\x<k2)-8;
}]

           \draw[name path=diag,black] (0,0) -- (10,10);


          \coordinate (g) at (g, {diag(g)});
           \coordinate (e1) at (e1, {diag(e1)});
            \coordinate (e2) at (e2, {diag(e2)});
             \coordinate (s) at (s, {diag(s)});

             \coordinate (a) at (g, {diag(f)});
              \coordinate (b) at (f2, {diag(g)});
               \coordinate (c) at (f2, {diag(f)});

             \draw[very thick, blue, name path=tu1] (e1) to[out=75,in=180] (g);
             \draw[very thick, blue, name path=td1] (e1) to[out=315,in=180] (a);
              \draw[very thick, blue, name path=tu2] (e2) to[out=75,in=205] (10,{diag(10)});
              \draw[very thick, blue, name path=td2] (e2) to[out=300,in=140] (b);
              \draw[very thick, blue, name path=td3,] (c) to[out=300,in=180] (10,{diag(0)});

		 \draw [gray, densely dotted, very thick] (g) -- (b) -- (c) -- (a) -- (g);

                   \path[name path=base] (0,0) to (10,0);

                  \path [name path=lineA](b) -- (f2,10);
		\draw [gray, very thin, loosely dashed, name intersections={of=lineA and tu2}] (b) -- (intersection-1)%
		\pgfextra{\pgfgetlastxy{\macrox}{\macroy} \global\let\macroy\macroy};
         \coordinate (z) at (intersection-1);
         \path [name path=lineA](intersection-1) -- (10,\macroy);
         \draw [gray, very thin, loosely dashed, name intersections={of=lineA and diag}] (z) -- (intersection-1);
		\draw [black!30!green,very thick, pattern=dots, pattern color=black!30!green] (f2,g) -- (intersection-1) -- (f2,f) ;
		
		\node (k1k2)[scale=0.6, shape=circle, fill] at (k1,k2) {} ;
\node (name)[below,scale=1.5] at (k1+2.5,k2+2) {$(K_1,K_2)$} ;
\draw[-latex', black, thick] (k1+2,k2+1) to[out=270,in=0] (k1+0.5,k2);


	\draw[name path=a, black, thick] plot[domain=0:10, samples=100] (\x,{a(\x)}) ;
	\draw[name path=b, black, thick] plot[domain=0:10, samples=100] (\x,{b(\x)}) ;

                          \draw [gray, very thin, loosely dashed] (f,f) -- (a) ;
                        \draw [gray, very thin, loosely dashed] (f,f) -- (f,{b(f)}) ;
                       \draw [gray, very thin] (k2,k2) -- (k2,{b(k2)}) ;
                            \draw [gray, very thin] (k2,k2) -- (k1,k2) ;
                               \draw [gray, very thin] (k1,k2) -- (k1,{a(k1)}) ;
                                \draw [gray, very thin, loosely dashed] (a) -- (g,{b(20)-0.5}) ;
                                  \draw [gray, very thin, loosely dashed] (c) -- (f2,{a(10)-0.5}) ;
                                  \gettikzxy{(intersection-1)}{\m}{\n}
                                  \draw [gray, very thin, loosely dashed] (intersection-1) -- (\m,{a(\m)-0.5}) ;
                                  \node (n1aaa)[below,scale=1.5] at (\m,{a(\m)-0.5}) {$g(x^{\prime\prime})$} ;
                                  \node (n1aaaz)[below,scale=1.5] at (f2,{a(10)-0.5}) {$x^{\prime\prime}$} ;
                                  \coordinate (temp) at (intersection-1);
		\draw [gray, loosely dashed, very thin] (f,{b(f)}) -- (f,{a(10)-0.5}) ;
		\node (n1)[below,scale=1.5] at (f,{a(10)-0.5}) {$f^\prime$} ;
		\node (hedge)[left,black!30!green,scale=2] at (0,{b(0)}) {$\psi^{x^\prime,x^{\prime\prime}}$} ;
		\node (n2)[below,scale=1.5] at (g,{a(10)-0.5}) {$x^\prime=g^\prime$} ;
		
	\draw[red] plot[domain=0:f2, samples=100] (\x,{a(\x)+0.04});

       	
	\draw[black!30!green] plot[domain=0:f, samples=100] (\x,{b(\x)+0.02});
	\draw [black!30!green] (0,{b(0)+0.04}) -- (f,{b(f)+0.02}) ;
	\gettikzxy{(temp)}{\ax}{\ay}
	\gettikzxy{(g)}{\u}{\i}

	\def\ra{( a(f2)-a(\ax) ) *((\ax-f2)^(-1))};
	\draw [black!30!green,very thick]  (\ax,{a(\ax)}) -- (\u, {(  550*\ra +a(\ax)+0.06});
	\draw [black!30!green,very thick]  (f,{b(f)}) -- (\u, {(  550*\ra +a(\ax)+0.06});

	\draw [black!30!green,very thick] (\m,{a(\m)+0.04})-- (10,{a(10)+0.04});
	
	\path [name path=lineG](temp) -- (10,\macroy);

\end{tikzpicture}
}
\caption{Picture of $f$ and $g$ along with superhedge for the dotted region $\sG$. The hedge function $\psi^{x',x''}$ has a kink at $x'$.}
\label{fig:greenRegion}
\end{figure}
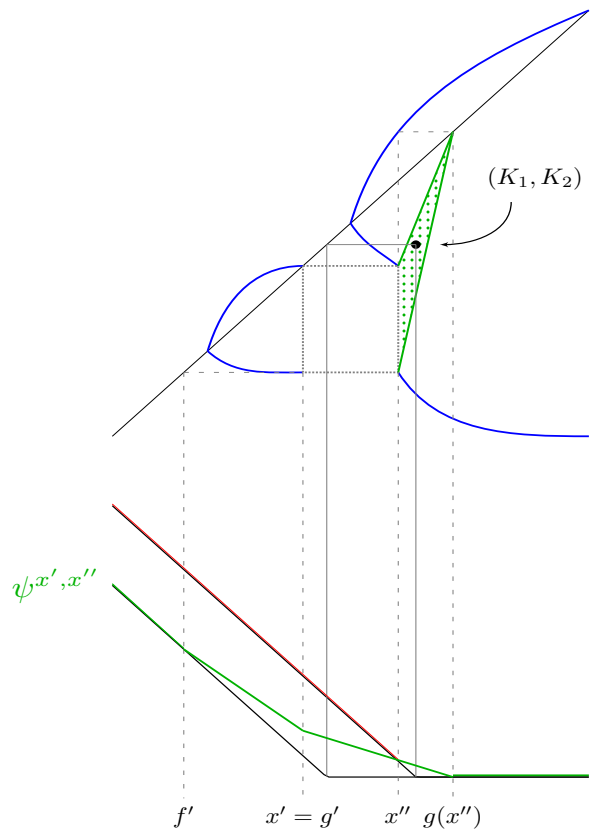

By construction $\psi^{x^\prime,x^{\prime\prime}}$ is convex and $\psi^{x^\prime,x^{\prime\prime}}(z)\geq (K_2-z)^+$ (see Figure \ref{fig:greenRegion}), and thus by Lemma~\ref{lem:phifrompsi} it can be used to construct a superhedge $(\psi^{x^\prime,x^{\prime\prime}}, \phi^{x^\prime,x^{\prime\prime}},\theta_{1,2}^{x^\prime,x^{\prime\prime}})$.

\begin{thm}
Suppose Assumption~\ref{ass:bimodal} holds and $(K_1,K_2) \in \sG$.
The model $M^{x^\prime,x^{\prime\prime}}$ and the stopping time $\tau=1$ if $X<x^{\prime\prime}$ and $\tau=2$ otherwise attains the highest consistent model price. Moreover, $\psi^{x^\prime,x^{\prime\prime}}$ defined in \eqref{eq:psiPrime} generates the cheapest superhedge and the highest model-based price is equal to the cost of the cheapest superhedge.

\label{thm:bimodal}
\end{thm}

\begin{proof}
Under the candidate model $M^{x^\prime,x^{\prime\prime}}$ mass in $(f^\prime,x^\prime)$ at time 1 is mapped to the same interval at time 2, while the mass in $(x^\prime,x^{\prime\prime})$ is mapped to $(x^\prime,g(x^{\prime\prime}))$. Then under the candidate stopping time (exercise at time $1$ if $X<x^{\prime\prime}$ and at time $2$ otherwise) the law of $Y$ (under $M^{x^\prime,x^{\prime\prime}}$) on the event that the option was not exercised at time $1$ is given by $(\nu-\mu)\lvert_{(-\infty,f^\prime)} + \nu\lvert_{(g(x^{\prime\prime}),\infty)}$. Therefore
\begin{align*}
MBEP&= \int_{-\infty}^{x^{\prime\prime}} (K_1-w)^+ \mu(dw) + \int_{-\infty}^{f^\prime} (K_2 - w)^+ (\nu - \mu)(dw)\\
& =  P_\mu(x^{\prime\prime}) +(K_1-x^{\prime\prime})P_{\mu}^\prime(x^{\prime\prime})+ D(f^\prime) +(K_2-f^\prime)D^\prime(f^\prime).
\end{align*}

Now consider the hedging cost generated by $\psi^{x^\prime,x^{\prime\prime}}$. Let $\Theta_1=\frac{K_2-f^\prime-\Delta_2(x^\prime)}{x^\prime-f^\prime}= - \Delta_1'$ and $\Theta_2=\frac{K_1-x^{\prime\prime}}{g(x^{\prime\prime})-x^{\prime\prime}}= - \Delta_2'$. Note that we can rewrite (\ref{eq:psiPrime}) as
$\psi^{x^\prime,x^{\prime\prime}}(z)=\Theta_2(g(x^{\prime\prime})-z)^+ + (\Theta_1-\Theta_2)(x^\prime-z)^+ +(1-\Theta_1)(f^\prime-z)^+$ .
Then
\begin{equation*}
\phi(z)=(1-\Theta_1)[(x^{\prime}-z)^+-(f^\prime-z)^+] +(1-\Theta_2)[(x^{\prime\prime}-z)^+-(x^\prime-z)^+] ,
\end{equation*}
and thus the hedging cost is
\begin{eqnarray*}
HC & = &\Theta_2P_\nu(g(x^{\prime\prime})) + (1-\Theta_1)D(f^\prime) + (1-\Theta_2)P_\mu(x^{\prime\prime}) + (\Theta_1-\Theta_2)D(x^\prime)\\
& = & P_\mu(x^{\prime\prime}) + D(f^\prime) +\Theta_1[D(x^\prime)-D(f^\prime)]+\Theta_2[P_\nu(g(x^{\prime\prime}))-P_\nu(x^\prime)-P_\mu(x^{\prime\prime})+P_\mu(x^\prime)].
\end{eqnarray*}
Now using (\ref{eq:massmeanBimodal}) and the fact that $g(x')=x'$ we have that $D'(f')=D'(x')$ and $f'D'(f') - D(f') = x'D'(x')- D(x')$. Hence
\begin{equation}\label{eq:temp1}
\Theta_1[D(x^\prime)-D(f^\prime)]=(K_2-f^\prime)D^\prime(f^\prime)-\Delta_2(x^\prime)D^\prime(f^\prime);
\end{equation}
moreover (\ref{eq:massmeanx''}) gives that
\begin{align}
\Theta_2[P_\nu(g(x^{\prime\prime})) & - P_\nu(x^\prime)-P_\mu(x^{\prime\prime})+P_\mu(x^\prime)]\nonumber\\
& = \Theta_2[ g(x'') P'_\nu(g(x'')) - x'' P'_{\mu}(x'') - x' D'(x')] \nonumber \\
& = \Theta_2[ (g(x'') - x'') P'_\mu(x'') + (g(x'') - x')D'(x')]
\nonumber\\
&=(K_1-x^{\prime\prime})P^\prime_\mu(x^{\prime\prime})+\Delta_2(x^\prime)D^\prime(f^\prime).\label{eq:temp2}
\end{align}
Then, combining (\ref{eq:temp1}) and (\ref{eq:temp2}) we conclude that $HC=MBEP$.
\end{proof}

\subsubsection{$K_1 > r_\mu$}
\label{sssec:K1>rmu}

In Lemma~\ref{lem:LCdispersion}, and under the Dispersion Assumption, we constructed $f$ and $g$ but only on the interval $(e_-,r_\mu]$. More generally, under Standing Assumption~\ref{sa:noatoms} the arguments of Beiglb\"{o}ck and Juillet~\cite{BeiglbockJuillet:16} and Henry-Labord\`{e}re and Touzi~\cite{HenryLabordereTouzi:16} allow us to construct $T_d=f$ and $T_u=g$ on $[\ell_\mu,r_\mu]$ for arbitrary laws $\mu \leq_{cx} \nu$. For their purposes the definitions of $f$ and $g$ outside the range of $\mu$ are not important since they have no impact on the construction of the left-curtain martingale coupling.

Nonetheless, we can extend the definitions of $f$ and $g$ to $\R$ in a way which respects the conditions in Lemma~\ref{lem:LC}, by setting
\[ \begin{array}{ccl}
f(x)=x=g(x), & \hspace{10mm} & -\infty< x \leq \ell_\mu;   \\
f(x) = \ell_\nu, \hspace{7mm} g(x)=r_\nu, && r_\mu < x<r_\nu; \\
f(x)=x=g(x), && r_\nu \leq x < \infty.
\end{array}
\]
We will show that with these definitions for $f$ and $g$ the analysis of the previous sections extends to the case $K_1> r_\mu$.

Suppose $r_\nu>r_\mu$ and $r_\mu < K_1 < r_\nu$. Then $\Lambda(r_\mu) = \frac{r_\nu - K_1}{r_\nu - r_\mu} - \frac{(K_1-K_2)}{r_\mu - \ell_\nu}$ and $\Lambda(r_\nu-)=\infty$. If $\Lambda(r_\mu) \geq 0$ and $\Lambda$ is continuous then there exists $x^* \in [\ell_\mu,r_\mu]$ such that $g(x^*)>x^*$ and $\Lambda(x^*)=0$. Then, exactly as in Section~\ref{ssec:amput} we can construct a model, stopping time and superhedge such that the model-based expected payoff equals the hedging cost, and hence the model, stopping time and hedge are all optimal. The model could be based on the left-curtain coupling, and the optimal exercise rule is to exercise the American put at time 1 if $X < x^*$. Even if $\Lambda$ is not continuous, there may exist $x^*$ such that $\Lambda(x^*)=0$ and the same arguments apply (see Section~\ref{sssec:R}). If not, then we are in the setting of Section~\ref{sssec:G}, but again we can identify the optimal model and hedge. Essentially, the case $\Lambda(r_\mu) \geq 0$ is covered by a direct extension of existing arguments. Note that $\Lambda(r_\mu) \geq 0$ is equivalent to
\[ K_2 \geq K_1 - \frac{(r_\mu - \ell_\nu)(r_\nu - K_1)}{r_\nu - r_\mu}. \]

Now suppose $r_\mu < K_1 < r_\nu$ and $K_2 < K_1 - \frac{(r_\mu - \ell_\nu)(r_\nu - K_1)}{r_\nu - r_\mu}$. Then $\Lambda(r_\mu)<0$ and since $\Lambda(r_\nu-)=\infty$ and $\Lambda$ is continuous on $[r_\mu,r_\nu]$ (note that we have defined $f$ and $g$ to be constants on this range) there must exist $x^* \in (r_\mu,K_1)$ such that $\Lambda(x^*)=0$. It is always optimal to exercise at time 1 and any martingale coupling can be used to generate a model which attains the highest model based price of $P_\mu(K_1)=(K_1 - \overline{\mu})$. A cheapest superhedge is generated by
\begin{equation}
\label{eq:psidefK1}
 \psi(y) = \frac{K_2-\ell_\nu}{r_\nu-\ell_\nu} (r_\nu - y)^+ + \frac{r_\nu - K_2}{r_\nu-\ell_\nu} (\ell_\nu - y)^+.
\end{equation}
The cost of this hedge is
\begin{eqnarray*}
\lefteqn{ \frac{K_2-\ell_\nu}{r_\nu-\ell_\nu} P_\nu(r_\nu) + \frac{r_\nu - K_2}{r_\nu-\ell_\nu} P_\nu(\ell_\nu) + P_\mu(K_1)
         - \frac{K_2-\ell_\nu}{r_\nu-\ell_\nu} P_\mu(r_\nu) - \frac{r_\nu - K_2}{r_\nu-\ell_\nu} P_\mu(\ell_\nu) } \\
         & = & \frac{K_2-\ell_\nu}{r_\nu-\ell_\nu} (r_\nu - \bar{\mu})+ (K_1 - \bar{\mu}) - \frac{K_2-\ell_\nu}{r_\nu-\ell_\nu}(r_\nu - \bar{\mu}) = (K_1 - \bar{\mu}).
         \end{eqnarray*}


Finally suppose $K_1>r_\nu$. Then $X<K_1$ almost surely under any consistent model and for $X<K_1$
\[ \E[(K_2-Y)^+|X] < E[(K_1-Y)^+|X] = (K_1-X).  \]
It is always optimal to exercise the American put at time 1. If $K_2> r_\nu$ or $K_2< \ell_\nu$ then we are in the case studied in Section~\ref{sssec:B} and a cheapest hedge is generated by a time 2 payoff $\psi(y) = (K_2-y)^+$. If $K_2 \in [\ell_\nu, r_\nu]$ then we are in the case studied in Section~\ref{sssec:W} and a cheapest superhedge is generated by $\psi =\psi(y)$ where $\psi$ is given by \eqref{eq:psidefK1}. In either case
the highest model-based expected payoff is $P_\mu(K_1)= (K_1 - \bar{\mu})$ and this is also the cost of the superhedge.

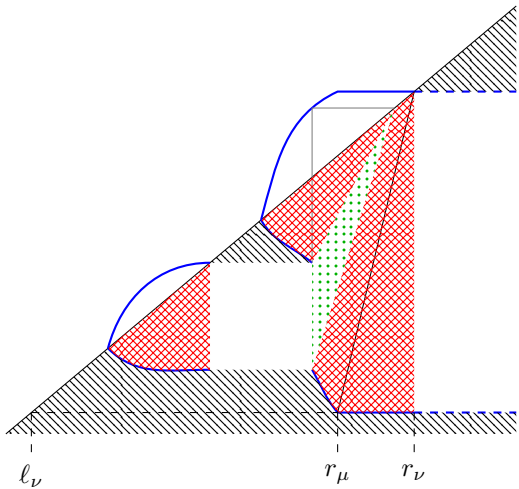
\begin{figure}[H]
\centering
\begin{tikzpicture}[
declare function={	
diag(\x)=\x;
    	f=1.5;
	e1=2;
	g=4;
	e2=5;
	f2=6;
	s=8.5;
}]
\begin{axis}[axis lines=middle,
	ytick=\empty,
            xtick=\empty,
             xmin=-.1, xmax=10,
             ymin=0, ymax=10,
             yticklabels={},
            xticklabels={},
            axis line style={draw=none}, clip=false]

           \addplot[name path=diag,black,domain={0:10}] {diag(\x)};

                \path[name path=base] (0,0) to (10,0);


          \coordinate (g) at (g, {diag(g)});
           \coordinate (e1) at (e1, {diag(e1)});
            \coordinate (e2) at (e2, {diag(e2)});
             \coordinate (s) at (s, {diag(s)});

             \coordinate (a) at (g, {diag(f)});
              \coordinate (b) at (f2, {diag(g)});
               \coordinate (c) at (f2, {diag(f)});

             \draw[thick, blue, name path=tu1] (e1) to[out=75,in=180] (g);
             \draw[thick, blue, name path=td1] (e1) to[out=315,in=180] (a);
              \draw[thick, blue, name path=tu2] (e2) to[out=75,in=205] (6.5,8) -- (8,8);
               \draw[thick,dashed, blue, name path=tu2x] (8,8)--(10,8);
              \draw[thick, blue, name path=td2] (e2) to[out=300,in=140] (b);
              \draw[thick, blue, name path=td3,] (c) to[out=300,in=130] (6.5,0.5)--(8,0.5);
                  \draw[thick, dashed,blue, name path=td3x,] (8,0.5)--(10,0.5);
              
                \draw[black] (6.5,0.5)-- (8,8);

                \path[gray, name path=dottop] (g) to (b);
                 \path[gray, name path=dotright] (b) to (c);
                  \path[gray, name path=dotbottom] (a) to (c);
                   \path[gray, name path=dotleft] (g) to (a);

		\path [gray, dotted, pattern=dots, pattern color=white] (g) -- (b) -- (c) -- (a) -- (g);

                  \path [name path=lineA](b) -- (f2,10);
		\draw [gray, name intersections={of=lineA and tu2}] (b) -- (intersection-1)%
		\pgfextra{\pgfgetlastxy{\macrox}{\macroy} \global\let\macroy\macroy};
         \coordinate (z) at (intersection-1);
         \path [name path=lineA](intersection-1) -- (10,\macroy);
         \draw [gray,name intersections={of=lineA and diag}] (z) -- (intersection-1);

		\path [thick, green, pattern=dots, pattern color=black!30!green] (f2,g) -- (intersection-1) -- (f2,f) -- (f2,g);

		
		 \path [pattern=north west lines, pattern color=black] (0,0) -- (e1) to[out=315,in=180] (a) -- (c) to[out=300,in=130] (6.5,0.5) -- (10,0.5)--(10,0)--(0,0);
	
		\path [ pattern=north west lines, pattern color=black] (g) -- (e2) to[out=300,in=140] (b) -- (g);
		
		\path [ pattern=north west lines, pattern color=black] (8,8) -- (10,8)--(10,10) -- (8,8);

		
                            \path [pattern=crosshatch,pattern color=red] (b) -- (intersection-1) -- (e2) to[out=300,in=140] (b);
                            \path [pattern=crosshatch,pattern color=red] (8,0.5) -- (8,8)-- (intersection-1) -- (c) to[out=300,in=130] (6.5,0.5)--(8,0.5);

                	 \path [pattern=crosshatch,pattern color=red] (e1) to[out=315,in=180] (a) -- (g) -- (e1);

  \draw [black, thin, dashed] (0.5,0.5) -- (0.5,-0.5) ;
    \draw [black, thin, dashed] (0.5,0.5) -- (6.5,0.5) ;
                                  \node (0.5,-0.5)[below,scale=1] at (0.5,-0.5){$\ell_\nu$} ;
                                  
                                    \draw [black, thin, dashed] (6.5,0.5) -- (6.5,-0.5) ;
                                    \node (6.5,-0.5)[below,scale=1] at (6.5,-0.5){$r_\mu$} ;
                                    
                                     \draw [black, thin, dashed] (8,0.5) -- (8,-0.5) ;
                                    \node (8,-0.5)[below,scale=1] at (8,-0.5){$r_\nu$} ;

             \end{axis}
\end{tikzpicture}
\caption{The various cases for $K_1 > r_\nu$ in the setting of Section~\ref{ssec:bimodal}.}
\label{fig:K1big}
\end{figure}

\subsection{Intervals where $\nu$ has no mass, or $\nu = \mu$.}
\label{ssec:nunomass}
The definition of the left-curtain martingale coupling (recall Lemma~\ref{lem:FGconstruction}) only requires that $g=T_u$ is increasing, and not that it is continuous. In general $g$ may have jumps; such jumps occur when there is an interval on which $\nu$ places no mass.

If $g$ has a jump then we need to adapt the superhedge.
Suppose $g$ has a jump at $\hat{x}$ (which has to be upwards since $g$ is increasing), and suppose $f$ is continuous at $\hat{x}$. Suppose further that $K_1$ is such that $\hat{x}\in(g^{-1}(K_1),K_1)$. Then as before, we would like to find $x^*\in(g^{-1}(K_1),K_1)$ such that $\Lambda(x^*)=0$. Recall that $\Lambda$ is increasing and suppose $\Lambda(g^{-1}(K_1))<0<\Lambda(K_1)$. If $\Lambda(\hat{x}-)<0$ and $\Lambda(\hat{x}+)>0$, then there will be no solution to $\Lambda=0$. On the other hand, by keeping $x=\hat{x}, \hat{f}=f(\hat{x})$ fixed in (\ref{eq:Lambdadef}), and varying $g$ only, we see that there exists $\hat{g}\in(g(\hat{x}-),g(\hat{x}+))$, such that $(\hat{g}-K_1)/(\hat{g}-\hat{x})=(K_1-K_2)/(\hat{x}-\hat{f})$ so that $\Upsilon(f(\hat{x}), \hat{x}, \hat{g}) = 0$. Then, the candidate (and indeed optimal) superhedging strategy is generated by $\psi^*$, given in (\ref{eq:psi*}), with $(f^*,x^*,g^*)=(\hat{f},\hat{x},\hat{g})$, see Figure \ref{fig:NUatomsJumps}. Moreover, since $\nu$ does not charge $(g(\hat{x}-),g(\hat{x}+))$, the triple $(\hat{f},\hat{x},\hat{g})$ solves the mass and mean equations (\ref{eq:mass}) and (\ref{eq:mean}). The strong duality between the model-based expected payoff and the hedging cost follows as before.

\begin{figure}[H]
\centering
\begin{tikzpicture}[scale=0.5,
    declare function={	
    	k1=10;
	k2=5;
	f=3;
	g=12;
	g0=19;
	z=9;
	g1=f+(k2-f)*(z-f)/(k2-f-k1+z);
	a(\x)=(k1-\x)*(\x<k1);
	b(\x)=(k2-\x)*(\x<k2);
		}]

	\draw[name path=a, black, thick] plot[domain=0:20, samples=100] (\x,{a(\x)});
	\draw[name path=b, black, thick] plot[domain=0:20, samples=100] (\x,{b(\x)}) ;
	
	\draw[name path=h1, red, thick] plot[domain=f:z, samples=100] (\x,{b(f)-((b(f)-a(z))/(z-f))*(\x-f)}) ;
	\draw[name path=h2, red, thick] plot[domain=z:g1, samples=100] (\x,{a(z)-(a(z)/(g1-z))*(\x-z)});
	\draw[name path=h3, blue, thick,dashed] plot[domain=z:g, samples=100] (\x,{a(z)-(a(z)/(g-z))*(\x-z)});
	\draw[name path=h4, blue, thick,dashed] plot[domain=z:g0, samples=100] (\x,{a(z)-(a(z)/(g0-z))*(\x-z)});

\node (strike1) at (k1,-0.5) {$K_1$};
\node (strike2) at (k2,-0.5) {$K_2$};
\node (f) at (f,-0.5) {$\hat{f}$};
\node (z) at (z,-0.5) {$\hat{x}$};
\node (g) at (g,-0.5) {$g(\hat{x}-)$};
\node (g0) at (g0,-0.5) {$g(\hat{x}+)$};
\node (g1) at (g1,-0.5) {$\hat{g}$};

\path [name path=lineA](f,0)--(f,11);
\draw [name intersections={of=lineA and b}, black, dashed] (f,0) -- (intersection-1);

\path [name path=lineA](z,0)--(z,11);
\draw [name intersections={of=lineA and a}, black, dashed] (z,0) -- (intersection-1);

\end{tikzpicture}
\caption{Sketch of put payoffs with points $\hat{x}$, $\hat{f}$ and $\hat{g}$ marked.}
\label{fig:NUatomsJumps}
\end{figure}
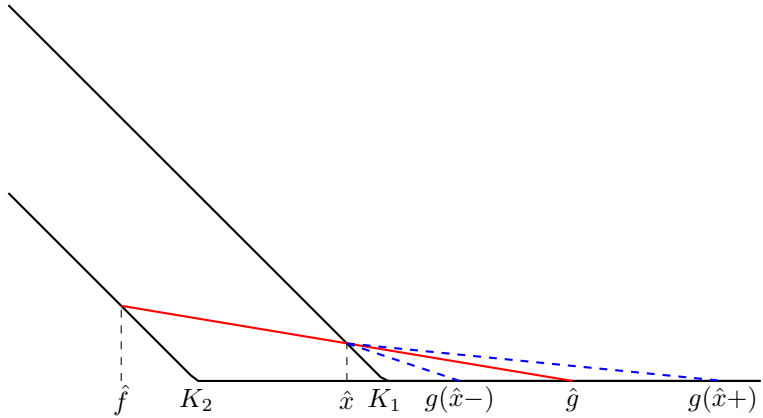

Alternatively, suppose $f$ has a downward jump at $\bar{x}$. This can happen if $\nu=\mu$ on $(f(\bar{x}+),f(\bar{x}-))$. Suppose that $K_1$ is such that $\bar{x}\in(g^{-1}(K_1),K_1)$ and $\Lambda(\bar{x}-)<0$ and $\Lambda(\bar{x}+)>0$, so that again we cannot find $x\in(g^{-1}(K_1),K_1)$ with $\Lambda(x)=0$. We can deal with this similarly as in the case of discontinuity in $g$: choose $\bar{f}\in(f(\bar{x}+),f(\bar{x}-))$ such that $\Upsilon(\bar{f},\bar{x},g(\bar{x}))=0$, then consider a hedging strategy generated by $\psi^*$ with $(f^*,x^*,g^*)=(\bar{f},\bar{x},g(\bar{x}))$. Note that $\mu = \nu$ on $(f(\bar{x}+),f(\bar{x}-))$ and so if (\ref{eq:mass}) and (\ref{eq:mean}) hold for some $f \in [f(\bar{x}+),f(\bar{x}-)]$ (with $\bar{x}, \bar{g}$) then they hold for all $f$ in this interval.  It follows that we can construct a coupling in which
 $(\bar{f},\bar{x})$ is mapped to $(\bar{f},\bar{g})$ and strong duality holds.

For both jumps in $f$ and $g$, we have a pictorial representation of the regions of pairs $(K_1,K_2)$ which lead to a hedging strategy which has to be adapted as above, see Figure \ref{fig:nuAtoms}. If $g$ has a jump at $\hat{x}$, then $\Lambda(\hat{x}-)<0$ and $\Lambda(\hat{x}+)>0$ is equivalent to point $(K_1,K_2)$ lying in the interior of a triangle with vertices $\{(g(\hat{x}-),g(\hat{x}-)),(g(\hat{x}+),g(\hat{x}+)),(\hat{x},f(\hat{x}))\}$. On the other hand, if $f$ jumps downwards at $\bar{x}$, then $\Lambda(\bar{x}-)<0$ and $\Lambda(\bar{x}+)>0$ is equivalent to point $K_1$, $K_2$ lying in the interior of a triangle with vertices $\{(\bar{x},f(\bar{x}-))$, $(\bar{x},f(\bar{x}+))$, $(g(\bar{x}),g(\bar{x}))\}$ (compare this with a region $\sG$).

Exceptionally we may have simultaneous jumps in $g$ and $f$ at $\check{x}$. Then the set of $(K_1,K_2)$ for which these arguments are needed is a quadrilateral with vertices $(\check{x}, f(\check{x}-))$, $(\check{x}, f(\check{x}+))$, $(g(\check{x}+), g(\check{x}+))$ and $(g(\check{x}-), g(\check{x}-))$. In particular, then there are multiple pairs $(\check{f},\check{g})$ with $\check{f}\in(f(\check{x}+), f(\check{x}-))$ and $\check{g}\in(g(\check{x}-), g(\check{x}+))$ such that $\Upsilon(\check{f}, \check{x}, \check{g}) = 0$, so that an optimal hedging strategy is not unique.

\subsection{The general case for continuous $\nu$}
\label{ssec:generalcts}
In the previous sections we showed how the left-curtain coupling can be used to find an optimal model, exercise strategy and a superhedge, under the assumption that both $\mu$ and $\nu$ are continuous together with further regularity and simplifying assumptions which we labelled the Dispersion Assumption and the Single Jump Assumption.
Under the latter assumption, the existence of points that solve (\ref{eq:massmeanBimodal}) led us to identify two further types of hedging strategy that were not present under the dispersion assumption, making four in total.

If we relax the assumptions further and require only that both $\mu$ and $\nu$ are continuous, then we expect that there exist multiple pairs $(f^\prime_i,x^\prime_i)$, $i=1,2,3,...$, that solve (\ref{eq:massmeanBimodal}).
Note that from the monotonicity of $g$ we can write $\{ x : g(x) > x \}$ as a countable union of intervals, and on each such interval $f$ is decreasing. $f$ jumps over the intervals $(f_i',x_i')$ identified above (at least those with $x'$ to the left of the current value of $x$). In particular, $f$ has only countably many downward jumps. Figure~\ref{fig:generalFG} is a stylized representation of the general left-curtain martingale coupling, not least because in the figure $f$ has only finitely many jumps. Using Figure~\ref{fig:generalFG} we can divide $(K_1, K_2<K_1)$ into four regions, see Figure~\ref{fig:generalRegions}. They key point is that these four regions are characterised exactly as in the cases described in Section~\ref{ssec:bimodal}. For given $(K_1,K_2)$ we can determine which of the types of hedging strategy is a candidate optimal superhedge, and determine a candidate optimal stopping rule. (We can always use the model associated with the left-curtain martingale coupling $\pi_{lc}$.) The fact that these candidates are indeed optimal can be proved using exactly analogous techniques to those used in Section~\ref{ssec:bimodal}.

\begin{figure}[H]

\centering
\resizebox{11 cm}{9cm}{
\begin{tikzpicture}[dot/.style={circle,inner sep=1pt, fill, label={#1},name#1},
extended line/.style={shorten >=-#1, shorten <=-#1},
extended line/.default=3cm,
declare function={	
diag(\x)=\x;
    	f=1;
	e1=1.5;
	g=3;
	e2=3.5;
	f2=4;
	f3=5;
	fx=5.5;
	e1x=6;
	gx=7;
	e2x=7.5;
	f2x=8;
	f3x=8.5;
	s=8.5;
	sx=8.5;
	k1=6.37;
	k2=4.5;
	a(\x)=(k1-\x)*(\x<k1)-8;
	b(\x)=(k2-\x)*(\x<k2)-8;
}]

           \draw[name path=diag,black] (0,0) -- (10,10);


          \coordinate (g) at (g, {diag(g)});
           \coordinate (e1) at (e1, {diag(e1)});
            \coordinate (e2) at (e2, {diag(e2)});
             \coordinate (s) at (s, {diag(s)});

             \coordinate (a) at (g, {diag(f)});
              \coordinate (b) at (f2, {diag(g)});
               \coordinate (c) at (f2, {diag(f)});

             \draw[thick, blue, name path=tu1] (e1) to[out=75,in=180] (g);
             \draw[thick, blue, name path=td1] (e1) to[out=315,in=180] (a);
              \draw[thick, blue, name path=tu2] (e2) to[out=75,in=205] (f3,{diag(f3)});
              \draw[thick, blue, name path=td2] (e2) to[out=300,in=140] (b);
              \draw[thick, blue, name path=td3,] (c) to[out=300,in=180] (f3,0.5);

               \coordinate (gx) at (gx, {diag(gx)});
           \coordinate (e1x) at (e1x, {diag(e1x)});
            \coordinate (e2x) at (e2x, {diag(e2x)});
             \coordinate (sx) at (sx, {diag(sx)});

             \coordinate (ax) at (gx, {diag(fx)});
              \coordinate (bx) at (f2x, {diag(gx)});
               \coordinate (cx) at (f2x, {diag(fx)});
               \draw[thick, blue, name path=tu1x] (e1x) to[out=75,in=180] (gx);
             \draw[thick, blue, name path=td1x] (e1x) to[out=315,in=180] (ax);
              \draw[thick, blue, name path=tu2x] (e2x) to[out=75,in=205] (10,10);
              \draw[thick, blue, name path=td2x] (e2x) to[out=300,in=140] (bx);
              \draw[thick, blue, name path=td3x] (cx) to[out=300,in=180] (f3x,f3);
               \draw[thick, blue, name path=td4] (f3x,0.5) to[out=300,in=180] (10,0);

		 \draw [blue, dashed, pattern=dots, pattern color=white] (g) -- (b) -- (c) -- (a) -- (g);
		  \draw [blue, dashed, pattern=dots, pattern color=white] (gx) -- (bx) -- (cx) -- (ax) -- (gx);
		  \draw [blue, dashed, pattern=dots, pattern color=white] (f3,{diag(f3)}) -- (f3x,f3) -- (f3x,0.5) -- (f3,0.5) -- (f3,{diag(f3)});

                   \path[name path=base] (0,0) to (10,0);

                  \path [name path=lineA](b) -- (f2,10);
		\draw [gray, very thin, name intersections={of=lineA and tu2}] (b) -- (intersection-1)%
		\pgfextra{\pgfgetlastxy{\macrox}{\macroy} \global\let\macroy\macroy};
         \coordinate (z1) at (intersection-1);
         \path [name path=lineA](intersection-1) -- (10,\macroy);
         \draw [gray, very thin, name intersections={of=lineA and diag}] (z1) -- (intersection-1);
          \coordinate (z1) at (intersection-1);
		\draw [gray] (f2,g) -- (intersection-1) -- (f2,f) ;
		\path [thick, black!30!green, pattern=dots, pattern color=black!30!green] (f2,g) -- (intersection-1) -- (f2,f) -- (f2,g);

                  \path [name path=lineA](bx) -- (f2x,10);
		\draw [gray, very thin, name intersections={of=lineA and tu2x}] (bx) -- (intersection-1)%
		\pgfextra{\pgfgetlastxy{\macrox}{\macroy} \global\let\macroy\macroy};
         \coordinate (z2) at (intersection-1);
         \path [name path=lineA](intersection-1) -- (10,\macroy);
         \draw [gray, very thin, name intersections={of=lineA and diag}] (z2) -- (intersection-1);
         \coordinate (z2) at (intersection-1);
		\draw [gray] (f2x,gx) -- (intersection-1) -- (f2x,fx) ;
		 \path [very thick, green, pattern=dots, pattern color=black!30!green] (f2x,gx) -- (intersection-1) -- (f2x,fx) -- (f2x,gx);

                  \path [name path=lineA](f3x,{diag(f3)}) -- (f3x,10);
		\draw [gray, very thin, name intersections={of=lineA and tu2x}] (f3x,{diag(f3)}) -- (intersection-1)%
		\pgfextra{\pgfgetlastxy{\macrox}{\macroy} \global\let\macroy\macroy};
         \coordinate (z3) at (intersection-1);
         \path [name path=lineA](intersection-1) -- (10,\macroy);
         \draw [gray, very thin, name intersections={of=lineA and diag}] (z3) -- (intersection-1);
         \coordinate (z3) at (intersection-1);
		\draw [gray] (f3x,{diag(f3)}) -- (intersection-1) -- (f3x,0.5) ;
		\path [thick, green, pattern=dots, pattern color=black!30!green] (f3x,{diag(f3)}) -- (intersection-1) -- (f3x,0.5) -- (f3x,{diag(f3)});

		
		 \path [pattern=north west lines, pattern color=black] (0,0) -- (e1) to[out=315,in=180] (a) -- (c) to[out=300,in=180] (f3,0.5) -- (f3x,0.5) to[out=300,in=180] (10,0);
	
		\path [ pattern=north west lines, pattern color=black] (g) -- (e2) to[out=300,in=140] (b) -- (g);
		\path [pattern=north west lines, pattern color=black] (f3,f3) -- (e1x) to[out=315,in=180] (ax) -- (cx) to[out=300,in=180] (f3x,f3) -- (f3,f3);
	
		\path [ pattern=north west lines, pattern color=black] (gx) -- (e2x) to[out=300,in=140] (bx) -- (gx);

		
                            \path [pattern=crosshatch,pattern color=red] (a) -- (g) -- (e1) to[out=315,in=180] (a);
			 \path [pattern=crosshatch,pattern color=red] (e2) to[out=300,in=140] (b) -- (z1) -- (e2);
			
			  \path [pattern=crosshatch,pattern color=red] (c) to[out=300,in=180] (f3,0.5) -- (f3,f3) -- (z1) -- (c);
			
			    \path [pattern=crosshatch,pattern color=red] (ax) -- (gx) -- (e1x) to[out=315,in=180] (ax);
			 \path [pattern=crosshatch,pattern color=red] (e2x) to[out=300,in=140] (bx) -- (z2) -- (e2x);
			
			  \path [pattern=crosshatch,pattern color=red] (cx) to[out=300,in=180] (f3x,f3) -- (f3x,f3) -- (z3) -- (z2) --(cx);
			
			   \path [pattern=crosshatch,pattern color=red] (f3x,0.5) to[out=300,in=180] (10,0) -- (10.,10) -- (z3) -- (f3x,0.5) ;

\end{tikzpicture}
}
\caption{General picture of $f,g$ with shading of regions. There remain 4 types of shading corresponding to 4 forms of optimal hedge.}
\label{fig:generalRegions}
\end{figure}
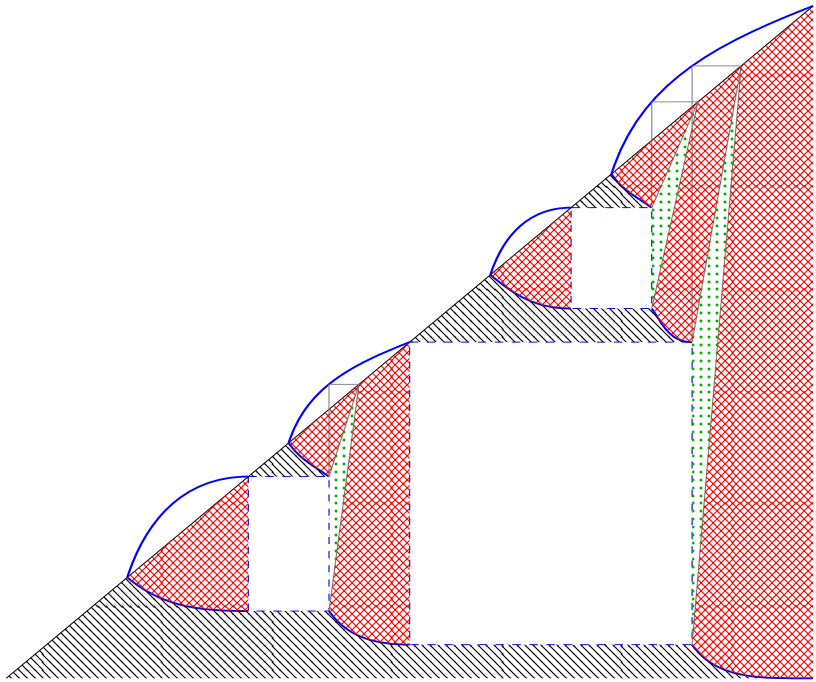

More specifically, we can divide $\{(k_1,k_2): k_2<k_1 \}$ into $\{ (k_1,k_2):k_2 \leq f(k_1)\} \cup \{(k_1,k_2): f(k_1) < k_2 < k_1 \}$. We can divide the former into two regions
$\sW = \{ (k_1,k_2) : K_2<k_1, \exists  x \leq k_1 \mbox{such that $f(x)<k_2<g(x)$} \}$ and $\sB = \{ (k_1,k_2): k_2 \leq f(k_1) \} \setminus \sW$. The latter we again divide into two regions $\sG$ and $\sR = \{ (k_1,k_2): f(k_1) < k_2 < k_1\} \setminus \sG$. Here we can write $\sG = \cup_{x:f(x-)>f(x+) } \Delta(x)$ where $\Delta(x)$ is a triangle with vertices $(x, f(x+))$, $(x, f(x-))$ and $(g(x),g(x))$. Then on each of the regions $\sW$, $\sB$, $\sG$ and $\sR$ we have a superhedge exactly as described in Section~\ref{ssec:bimodal}. Moreover, again by the arguments of Section~\ref{ssec:bimodal}, we can show that the hedging cost associated with the super-hedging strategy is precisely the model-based expected payoff of the American put under the martingale coupling $\pi_{lc}$ (and candidate stopping rule) thus proving the optimality of the hedge and of the model/exercise rule.

\begin{rem}
The set $\{x: g(x)>x\}$ is a collection of intervals and we let $I_+$ denote the set of right-endpoints of these intervals. As remarked above, Figure~\ref{fig:generalRegions} is drawn in the case of `finite complexity' in the sense that the set $I_+$ contains a finite number of elements. The results extend easily to countable $I_+$ provided $I_+$ contains no accumulation points.

In general $I_+$ may contain an accumulation point, and as discussed in Henry-Labord\`{e}re and Touzi~\cite{HenryLabordereTouzi:16}, care is needed in the construction of the left-curtain mappings $(T_d,T_u)$ in this case. However, from our perspective such subtleties do not cause a problem. The reason for this is we do not aim to derive the left-curtain coupling, but rather take the left-curtain coupling as a given, and use it to solve the put pricing problem.

Our construction of the best model and the cheapest hedge is local in the sense that when in Figure~\ref{fig:generalRegions} we look at in which region the point $(K_1,K_2)$ lies, the fine detail of the picture in other parts of $(k_1,k_2)$-space is not important. So, the existence of accumulation points can only be an issue if $K_1$ is equal to one of those accumulation points.

Let $x_\infty$ be such an accumulation point in $I_+$ and suppose $K_1=x_\infty$. Depending on the value of $K_2$ then either there exists $(x',f')$ with $f' < K_2 < x'$ such that \eqref{eq:massmeanBimodal} holds or not. In the former case we can follow the analysis of Section~\ref{sssec:W}, and in the latter Section~\ref{sssec:B}: in either case we construct a model and hedge such that the model price and hedging cost agree, thus proving optimality of both.
\end{rem}

\subsection{Atoms in the target law}
\label{ssec:atomsinnu}
When $\nu$ has atoms, the preservation of mass and mean conditions become (\ref{eq:massAtom}) and (\ref{eq:meanAtom}), respectively. In particular, atoms of $\nu$ correspond to the flat sections in $f$ or $g$. See Figure~\ref{fig:nuAtoms}.
In this case we still can find all the optimal quantities as before. In particular, $\Lambda(x):= \frac{g(x)-K_1}{g(x)-x} - \frac{(K_1 - K_2)}{x-f(x)}$ is strictly increasing in $x$, even if $f$ and/or $g$ is constant. Hence we can find solutions to $\Lambda=0$ (more generally solutions $x,f \in \aleph(x)$ to $\Upsilon(f,x,g=g(x))=0$) exactly as before. The superhedge is unchanged. A little care is needed in constructing the optimal model, but mass in $(f(x^*),x^*)$ is mapped to $(f(x^*),g(x^*))$ together with (potentially) atoms at $f(x^*)$ or $g(x^*)$. Specifically, given $f^*,x^*,g^*$ we can find $\lambda^*_f$ and $\lambda^*_g$ such that (\ref{eq:massAtom}) and (\ref{eq:meanAtom}) hold. Then, in any optimal model mass in $(f^*,x^*)$ is mapped to $\nu_{x^*}$ which is defined to be $\nu_{x^*} = \nu|_{(f^*,g^*)} + \lambda^*_{f^*} \delta_{f^*} + \lambda^*_{g^*} \delta_{g^*}$ and mass outside $(f^*,x^*)$ is mapped to $\nu - \nu_{x^*}$.

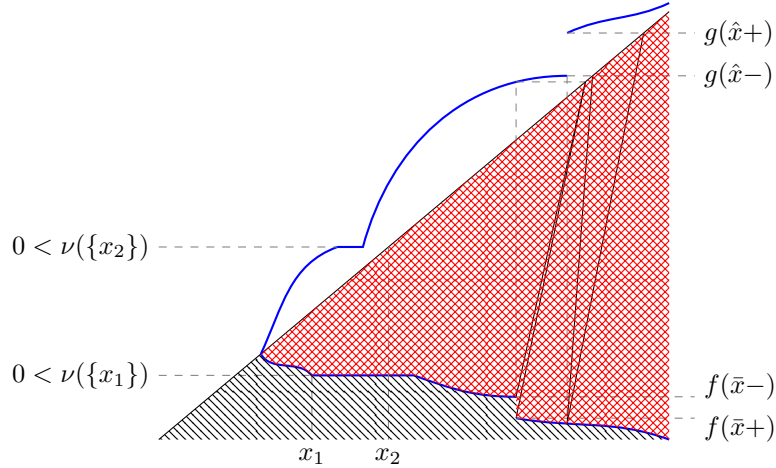
\begin{figure}[H]
\centering
\begin{tikzpicture}[
declare function={	
diag(\x)=\x;
    	e1=2;
	au=3.5;
	ad=3;
	ju=8;
	jd=7;
}]
\begin{axis}[axis lines=middle,
	ytick=\empty,
            xtick=\empty,
             xmin=-.1, xmax=10,
             ymin=0, ymax=10,
             yticklabels={},
            xticklabels={},
            axis line style={draw=none}, clip=false]

           \addplot[name path=diag,black,domain={0:10}] {diag(\x)};

                \path[name path=base] (0,0) to (10,0);


           \coordinate (e1) at (e1, {diag(e1)});

           \coordinate (au) at (au, {au+1});
            \coordinate (au1) at ({au+0.5}, {au+1});

              \coordinate (ad) at (ad, {e1-0.5});
            \coordinate (ad1) at ({ad+2}, {e1-0.5});

      \coordinate (ju) at (ju, {ju+0.5});
            \coordinate (ju1) at ({ju}, {ju+1.5});

            \coordinate (jd) at (jd, {ad-2});
            \coordinate (jd1) at ({jd}, {ad-2.5});

\draw[thick, blue, name path=tu1] (e1) to[out=65,in=200] (au) -- (au1) to[out=75,in=180] (ju);
 \draw[thick, blue, name path=tu2] (ju1) to[out=25,in=205] (10,{diag(10)+0.2});

             \draw[thick, blue, name path=td1] (e1) to[out=300,in=130] (ad) -- (ad1) to[out=340,in=180] (jd);
             \draw[thick, blue, name path=td2] (jd1) to[out=350,in=160] (10,{diag(0)});

             \draw[ thin, gray, dashed] (0,{e1-0.5}) -- (ad);
             \draw[ thin, gray, dashed] (0,{au+1}) -- (au);

              \draw[thin, gray, dashed] (ad) -- (ad,0);
		\node [below] at (ad,0) {$x_1$};
		
		      \draw[thin, gray, dashed] ({diag(au+1)},{diag(au+1)}) -- ({diag(au+1)},0);
              \node [below] at ({diag(au+1)},0) {$x_2$};

              	\node [left] at (0,{diag(au+1)}) {$0<\nu(\{x_2\})$};
		\node [left] at (0,{e1-0.5}) {$0<\nu(\{x_1\})$};
		
		 \draw[thin, gray, dashed] (ju) -- (10.5,{ju+0.5});
		  \draw[ thin, gray, dashed] (ju1) -- (10.5,{ju+1.5});
		
		  \node [right] at (10.5,{ju+0.5}) {$g(\hat{x}-)$};
		\node [right] at (10.5,{ju+1.5}) {$g(\hat{x}+)$};
		
		 \draw[ thin, gray, dashed] (jd) -- (10.5,{ad-2});
		  \draw[ thin, gray, dashed] (jd1) -- (10.5,{ad-2.5});
		
		  \node [right] at (10.5,{ad-1.75}) {$f(\bar{x}-)$};
		\node [right] at (10.5,{ad-2.75}) {$f(\bar{x}+)$};
		
			  \path [name path=lineA](ju) -- (ju,0);
		\draw [ very thin, gray, dashed, name intersections={of=lineA and td2}] (ju) -- (intersection-1);
		\draw[black] (intersection-1) -- ({ju+1.5},{ju+1.5});
		\draw[black] (intersection-1) -- ({ju+0.5},{ju+0.5});

            \path [name path=lineA](jd) -- (jd,10);
		\draw [very thin, gray, dashed, name intersections={of=lineA and tu1}] (jd) -- (intersection-1)%
		\pgfextra{\pgfgetlastxy{\macrox}{\macroy} \global\let\macroy\macroy};
         \coordinate (z) at (intersection-1);
         \path [name path=lineA](intersection-1) -- (10,\macroy);
         \draw [very thin, gray, dashed, name intersections={of=lineA and diag}] (z) -- (intersection-1);

         \draw[black] (intersection-1) -- (jd);
         \draw[black] (intersection-1) -- (jd1);

		
		 \path [pattern=north west lines, pattern color=black] (0,0) -- (e1) to[out=300,in=130] (ad) -- (ad1) to[out=340,in=180] (jd) -- (jd1) to[out=350,in=160] (10,{diag(0)})-- (0,0);
		
		
		 \path [pattern=crosshatch, pattern color=red] (e1) to[out=300,in=130] (ad) -- (ad1) to[out=340,in=180] (jd) -- (jd1) to[out=350,in=160] (10,{diag(0)}) -- (10,10) -- (e1);

             \end{axis}
\end{tikzpicture}
\caption{Atoms of $\nu$ correspond to flat sections in $f$ and $g$. Regions of no mass of $\nu$ correspond to jumps of $f$ and $g$.}
\label{fig:nuAtoms}
\end{figure}

\section{Discussion and extensions}
\label{sec:discussion}

\subsection{The role of the left-curtain coupling}
For any pair of strikes $(K_1,K_2)$ the left-curtain model attains the highest expected payoff for the American put. However, although it optimises simultaneously across all pairs of strikes it is not (in general) optimal for linear combinations of American puts. For example, if we consider a generalised American option with payoff $a$ if exercised at time 1 and $b$ if exercised at time 2, where $a(x) = \sum_{j=1}^J(K^j_{1}- x)^+$  and $b(y)=\sum_{j=1}^J(K_2^j - y)^+$ (with $K^j_2 \leq K^j_1$ for each $j$), then the model associated with the left-curtain coupling is typically not optimal. The reason is that a model $(\sS,M)$ is only optimal when it is combined with the best stopping rule, and the optimal stopping rule {\em does} depend on $(K_1, K_2)$.

Conversely, although the model associated with the left-curtain coupling is optimal (simultaneously across all pairs $K_1,K_2$), we do not need the full power of this coupling when we work with fixed $(K_1,K_2)$. In the dispersion assumption case all we need is a coupling in which $(f(x^*), x^*)$ is mapped to $(f(x^*),g(x^*))$ where $x^*$ is such that $\Lambda(x^*)=0$. There are many martingale couplings which have this property.

The intuition behind the optimality of the left-curtain coupling is as follows.
With American puts there is a tension between the time-decay of the option payout promoting early exercise, and the convexity of the payoff function promoting delay. If the aim is to maximise the payoff of the option then any paths which are in-the-money at time 1, and will remain in-the-money, are best exercised at time 1. However, once a path has been exercised, any further volatility is irrelevant. In particular, when designing a candidate optimal model we should try to keep paths which are exercised at time 1 constant (or near constant) whenever possible. Thus the probability space should be split into two regions: one region where the put is in-the-money at time 1 and is exercised, and thereafter paths move little, and a second region where the put is out-of-the-money at time 1 (and sometimes just in-the-money, but left unexercised at time 1) and then paths move a long way between times 1 and 2. The left-curtain coupling has this property.

\subsubsection{Multiple exercise times}
It is natural to ask if it is possible to extend the analysis to American puts which can be exercised at multiple dates $(T_1, T_2, \ldots T_N)$ where $N>2$, or equivalently to martingales $M = (M_n)_{0 \leq n \leq N}$ with marginals $(\mu_n)$ where $\mu_1$ has mean $M_0 = \bar{\mu}$ and $\mu_n \leq_{cx} \mu_{n+1}$ for $1 \leq n \leq N-1$. It is clear that many of the ideas extend naturally to the multi-marginal case. However, the number of types of hedging strategy may grow exponentially with $N$. This is left as future work.




\end{document}